\documentclass[11pt,a4paper]{article}

\usepackage[latin1]{inputenc}
\usepackage[T1]{fontenc}
\usepackage[danish,english]{babel} 

\usepackage{a4}
\usepackage[round]{natbib}

\usepackage{amsbsy}   
\usepackage{amsfonts}   
\usepackage{amsmath}         
\usepackage{amssymb}    
\usepackage{amsthm}   
\usepackage{thmtools}

\usepackage{enumitem}
\usepackage[hang,flushmargin]{footmisc} 

\usepackage{hyperref}
\usepackage{url}

\usepackage{appendix}
\usepackage{chngcntr}
\usepackage{etoolbox}

\usepackage{titlesec}
\titleformat{\chapter}[display]
{\bfseries\Large}
{\filcenter\MakeUppercase{\chaptertitlename} \Huge\thechapter}
{4ex}
{\titlerule
\vspace{2ex}%
\filcenter}
  [\vspace{2ex}%
\titlerule]
\titlespacing*{\chapter}{0pt}{-40pt}{40pt}

\usepackage{graphicx}
\usepackage{txfonts}
\usepackage{xcolor}

\usepackage[section]{placeins}

\allowdisplaybreaks

\usepackage{bibentry}
\nobibliography*

\newcommand{\xtl}{{X_{t_{i-1}^n}}}
\newcommand{\xtr}{{X_{t_{i}^n}}}
\newcommand{\wtl}{{W_{t_{i-1}^n}}}
\newcommand{\wtr}{{W_{t_{i}^n}}}
\newcommand{\ftl}{{\mathcal{F}_{t_{i-1}^n}}}
\newcommand{\thetan}{{\theta_0}}

\newcommand{\EE}{{\mathbb E}}

\newcommand{\PP}{{\mathbb P}}
\newcommand{\RR}{{\mathbb R}}
\newcommand{\NN}{{\mathbb N}}

\newcommand{\B}{{\mathbf B}}

\newcommand{\N}{{\mathbf N}}

\newcommand{\W}{{\mathbf W}}
\newcommand{\X}{{\mathbf X}}

\newcommand{\cc}{{\mathcal{C}}} 
\newcommand{\dd}{{\mathcal{D}}}
\newcommand{\ff}{{\mathcal{F}}}  
\renewcommand{\gg}{{\mathcal{G}}} 
\newcommand{\ii}{{\mathcal{I}}}  
\renewcommand{\ll}{{\mathcal{L}}} 
\newcommand{\nn}{{\mathcal{N}}}  
\newcommand{\pp}{{\mathcal{P}}}

\newcommand{\xx}{{\mathcal{X}}}

\declaretheorem[numberwithin=section]{theorem}
\declaretheorem[numberlike=theorem]{lemma}
\declaretheorem[numberlike=theorem]{corollary}
\declaretheorem[numberlike=theorem]{assumption}

\declaretheorem[numberlike=theorem, style = definition]{definition}

\declaretheorem[numberlike=theorem, style = remark]{remark}

\newcommand\xqed[1]{%
  \leavevmode\unskip\penalty9999 \hbox{}\nobreak\hfill
  \quad\hbox{#1}}
\newcommand\asqed{\xqed{$\diamond$}}

\newcommand\defqed{\xqed{$\diamond$}}
\newcommand\theoqed{\xqed{$\diamond$}}
\newcommand\reqed{\xqed{$\circ$}}

\numberwithin{equation}{section}
\numberwithin{theorem}{section}

\title{Efficient Estimation for Diffusions Sampled at High Frequency Over a Fixed
  Time Interval}

\author{Nina Munkholt Jakobsen \\ 
{\small Department of Mathematical Sciences \vspace{-1.5mm}} \\
{\small University of Copenhagen\vspace{-1.5mm}} \\
{\small Universitetsparken 5\vspace{-1.5mm}} \\
{\small DK-2100 Copenhagen {\O}\vspace{-1.5mm}} \\
{\small Denmark\vspace{-1.5mm}} \\ 
{\small munkholt@math.ku.dk}
\and
Michael S\o rensen$^*$ \\ 
{\small Department of Mathematical Sciences \vspace{-1.5mm}} \\
{\small University of Copenhagen\vspace{-1.5mm}} \\
{\small Universitetsparken 5\vspace{-1.5mm}} \\
{\small DK-2100 Copenhagen {\O}\vspace{-1.5mm}} \\
{\small Denmark\vspace{-1.5mm}} \\ 
{\small michael@math.ku.dk}}

\date{\today}

\begin{document} 

\maketitle

\begin{abstract}
Parametric estimation for diffusion processes is considered for high
frequency observations over a fixed time interval. The processes 
solve stochastic differential equations with an unknown parameter 
in the diffusion coefficient. We find easily verified conditions on approximate
martingale estimating functions under which estimators are consistent,
rate optimal, and efficient under high frequency (in-fill)
asymptotics. The asymptotic distributions of the estimators are shown to
be normal variance-mixtures, where the mixing distribution
generally depends on the full sample path of the diffusion process over 
the observation time interval. Utilising the concept of stable
convergence, we also obtain the more easily applicable
result that for a suitable data dependent normalisation, the estimators
converge in distribution to a standard normal distribution. The theory
is illustrated by a simulation study comparing an efficient and
a non-efficient estimating function for an ergodic and a non-ergodic model. 
{\bf Running title:} Efficient Estimation for High Frequency SDE Data. 
\end{abstract}

\newpage

\section{Introduction}
Diffusions given by stochastic differential equations
find application in a number of fields where they are used to
describe phenomena which evolve continuously in time. Some examples
include agronomy \citep{pedersen2000}, 
biology \citep{favetto2010}, finance \citep{merton1971, vasicek1977,
  CIR1985, larsen} and neuroscience \citep{ditlevsen2006, picchini2008, bibbona2010}.
\medskip

While the models have continuous-time dynamics, data are
only observable in discrete time, thus creating a demand for
statistical methods to analyse such data. With the exception of some
simple cases, the likelihood function is not explicitly known, and
a large variety of
alternate estimation procedures have been proposed in the
literature, see e.g.\ \citet{Helle2004} and
\citet{groenbog}. Parametric methods include the following. Maximum
li\-ke\-li\-hood-type estimation,  
primarily using Gaussian approximations to the likelihood
function, was considered by \citet{prao1983}, \citet{fz1989}, 
\citet{yoshida1992}, \citet{gc_jacod}, \citet{kessler_ergodic},
\citet{jacod2006}, \citet{gloter_ms} and
\citet{uchidayoshida}. Analytical expansions of the 
transition densities were investigated by \citet{asahalia2002,
  asahalia2008} and \citet{li2013}, while approximations to the score
function were studied by \citet{bibby1995}, \citet{kessler&ms}, \citet{jacobsen2001,
  jacobsen2002}, \citet{uchida2004}, and \citet{efficient}.
Simulation-based likelihood methods were developed by
\citet{pedersen1995}, \citet{roberts2001}, \citet{durham2002},
\citet{beskos2006, beskos2009}, \citet{golightly&wilkinson3, 
golightly&wilkinson4}, \citet{bladt2014}, and \citet{bladt15}.
\medskip


A large part of the parametric estimators proposed in the literature
can be treated within the framework of approximate martingale
estimating functions, see the review in \citet{MSbog}. In this paper,
we derive easily verified conditions on such estimating functions
that imply rate optimality and efficiency under a high frequency
asymptotic scenario, and thus contribute to providing clarity and a
systematic approach to this area of statistics.

\medskip

Specifically, the paper concerns parametric estimation for stochastic 
differential equations of the form
\begin{align}
dX_t=a(X_t)\, dt+b(X_t;\theta)\, dW_t\,,
\label{article1:SDE}
\end{align}
where $(W_t)_{t\geq 0}$ is a standard Wiener process. The drift and
diffusion coefficients $a$ and $b$ are deterministic functions,
and $\theta$ is the unknown parameter to be estimated. The drift
function $a$ needs not be known, but as examples in this paper show,
knowledge of $a$ can be used in the construction of estimating
functions. For ease of exposition,  $X_t$ and $\theta$ are both
assumed to be one-dimensional. The extension of our results to a
multivariate parameter is straightforward, and it is expected that
multivariate diffusions can be treated in a similar way. For $n\in \NN$, 
we consider observations $(X_{t_0^n}, X_{t_1^n},\ldots,X_{t_n^n})$ in
the time interval $[0,1]$, at discrete, equidistant time-points
$t_i^n = i/n,$ $i=0,1,\ldots,n$. We investigate the high frequency
scenario where $n\to \infty$. 
The choice of the time-interval $[0,1]$ is not restrictive since
results generalise to other compact intervals by suitable rescaling of
the drift and diffusion coefficients. The drift coefficient does not
depend on any parameter, because parameters that appear only in the
drift cannot be estimated consistently in our asymptotic scenario. 
\medskip

It was shown by \citet{dohnal1987} and  \citet{gobet2001} that under the
asymptotic scenario considered here, the model (\ref{article1:SDE}) is
locally asymptotic mixed normal with rate $\sqrt{n}$ and random
asymptotic Fisher information 
\begin{align}
\mathcal{I}(\theta) &= 2\int_0^1 \left( \frac{\partial_\theta b(X_s;
  \theta)}{b(X_s; \theta)}\right)^2\, ds.
\label{I}
\end{align}
Thus, a consistent estimator $\hat{\theta}_n$ is rate
optimal if $\sqrt{n}(\hat{\theta}_n - \thetan)$ converges in
distribution to a non-degenerate random variable as $n\to
\infty$, where $\thetan$ is the true parameter value.
The estimator is efficient if the limit may be written on the form
$\ii(\thetan)^{-1/2}Z$, where $Z$ is standard normal
distributed and independent of $\ii(\thetan)$. The concept of local
asymptotic mixed normality was introduced by \citet{jeganathan1982},
and is discussed in e.g. \citet[Chapter 6]{lecam2000} and
\citet{jacod2010}. 

\medskip

Estimation for the model (\ref{article1:SDE}) under the high
frequency asymptotic scenario described above was
considered by \citet{gc_jacod, gc1994}. These authors
proposed estimators based on a class of contrast
functions that were only allowed to depend on the observations
and the parameter through $b^2(\xtl; \theta)$ and
$\Delta_n^{-1/2}(\xtr-\xtl)$. The estimators were shown to be rate
optimal, and an efficient contrast function was identified. 
\cite{dohnal1987} gave estimators for particular cases of the model
(\ref{article1:SDE}). Apart from one instance, these estimators are
not of the type investigated by \citet{gc_jacod, gc1994}, but all apart
from one are covered by the theory in the present paper.
\medskip

In this paper, we investigate estimators based on the extensive class
of approximate martingale estimating functions
\begin{align*}
G_n(\theta)=\sum_{i=1}^n g(\Delta_n,\xtr,\xtl; \theta)\,
\end{align*}
with $\Delta_n = 1/n$, where the real-valued function $g(t,y,x; \theta)$ satisfies that
$\EE_\theta(g(\Delta_n,\xtr,\xtl; \theta)\mid \xtl)$ is of order
$\Delta_n^\kappa$ for some $\kappa \geq 2$. Estimators are obtained as
solutions to the estimating equation $G_n(\theta)=0$ and are referred
to as $G_n$-estimators. Exact martingale estimating functions, where
$G_n(\theta)$ is a martingale, constitute a particular case that is not
covered by the theory in \citet{gc_jacod, gc1994}. An example is the 
maximum likelihood estimator for the Ornstein-Uhlenbeck process with
$a(x) = -x$ and $b(x; \theta) = \sqrt{\theta}$, for
which $g(t,y,x;\theta) = (y-e^{-t}x)^2-\frac12 \theta (1-e^{-2t})$. A
simpler example of an estimating function for the same Ornstein-Uhlenbeck
process that is covered by our theory, but is not of the Genon-Catalot
\& Jacod-type, is given by $g(t,y,x;\theta) = (y-(1-t)x)^2-\theta t$. 

\medskip

The class of approximate martingale estimating functions was also studied by
\citet{efficient}, who considered high frequency observations in an
increasing time interval for a model like (\ref{article1:SDE}) where
also the drift coefficient depends on a parameter. Specifically,
the observation times were $t_i^n = i \Delta_n$ with $\Delta_n  \to 0$ 
and $n \Delta_n \to \infty$. Simple conditions on $g$ for rate
optimality and efficiency were found under the infinite horizon high frequency
asymptotics. To some extent, the methods of proof in the present paper are
similar to those in \citet{efficient}. However, while ergodicity of
the diffusion process played a central role in that paper, this property
is not needed here. Another important difference is that expansions of
a higher order are needed in the present paper, which complicates the proofs
considerably. Furthermore, the theory in the current paper requires a
more complicated version of the central limit theorem for martingales, and we need the
concept of stable convergence in distribution, in order to obtain
practically applicable convergence results. 

\medskip

First, we establish results on existence and uniqueness of consistent
$G_n$-estimators. We show that $\sqrt{n}(\hat{\theta}_n-\theta_0)$
converges in distribution to a normal variance-mixture,
which implies rate optimality. The limit distribution may be
represented by the product $W(\theta_0)Z$ of independent random
variables, where $Z$ is standard normal distributed. The
random variable $W(\theta_0)$ is generally non-degenerate, and
depends on the entire path of the diffusion process over the time-interval
$[0,1]$. Normal variance-mixtures
were also obtained as the asymptotic distributions of the
estimators of \citet{gc_jacod}. These distributions appear as
limit distributions in comparable non-parametric settings as well,
e.g. when estimating integrated volatility \citep{jacod_protter_1998,
  anova} or the squared diffusion coefficient \citep{fz1993,
  jacod2000kernel}. 
\medskip

Rate optimality is ensured by the condition that
\begin{align}
\partial_y g(0,x,x; \theta) = 0
\label{rate_opt_asymp}
\end{align}
for all $x$ in the state space of the diffusion process, and all parameter values
$\theta$. Here $\partial_y g(0,x,x; \theta)$ denotes the first
derivative of $g(0,y,x; \theta)$ with respect to $y$ evaluated in
$y=x$. The same condition was found in \citet{efficient} for rate
optimality of an estimator of the parameter in the diffusion
coefficient, and it is one of the conditions for small
$\Delta$-optimality; see \citet{jacobsen2001, jacobsen2002}. 
\medskip

Due to its dependence on $(X_s)_{s\in [0,1]}$, the limit distribution
is difficult to use for statistical applications, such as
constructing confidence intervals and test statistics. Therefore, we
construct a statistic $\widehat{W}_n$ that converges in
probability to $W(\theta_0)$. Using the stable convergence in
distribution of $\sqrt{n}(\hat{\theta}_n-\theta_0)$ towards
$W(\theta_0)Z$, we derive the more easily applicable result that
$\sqrt{n}\, \widehat{W}_n ^{-1}( \hat{\theta}_n -
  \theta_0)$ converges in distribution to a standard normal distribution.
\medskip

The additional condition that
\begin{align}
\partial^2_ y g(0,x,x; \theta) &=K_\theta \frac{ \partial_\theta b^2(x;
  \theta) }{b^4(x; \theta)}
\label{effcond}
\end{align}
($K_\theta\neq 0$) for all $x$ in the state space, and all parameter
values $\theta$, ensures efficiency of $G_n$-estimators. 
The same condition was obtained by
\cite{efficient} in his infinite horizon scenario for efficiency of
estimators of parameters in the diffusion coefficient. It is also
identical to a condition given by 
\citet{jacobsen2002} for small $\Delta$-optimality. The identity of the
conditions implies that examples of approximate martingale estimating
functions which are rate optimal and efficient in our asymptotic
scenario may be found in \citet{jacobsen2002} and \citet{efficient}.
In particular, estimating functions that are optimal in the sense of
\citet{gh1987} are rate optimal and efficient under weak regularity
conditions.

\medskip

\medskip

The paper is structured as follows: Section \ref{asymp:prelim}
presents definitions, notation and terminology used throughout the paper, as well as
the main assumptions. Section \ref{sec:asymp:main} states and
discusses our main results, while Section \ref{sec:asymp:sim} presents
a simulation study illustrating the results. Section
\ref{sec:mainlemmaproofs} contains main lemmas used to prove the main
theorem, and proofs of the main theorem and the lemmas.
Appendix \ref{sec:asymp:auxres} consists of auxiliary technical
results, some of them with proofs.
\section{Preliminaries}\label{asymp:prelim}

\subsection{Model and Observations}\label{asymp:setup:model}
Let $(\Omega,\ff)$ be a measurable space supporting a real-valued random
variable $U$, and an independent standard Wiener process $\W =
(W_t)_{t \geq 0}$. Let $(\ff_t)_{t\geq0}$ denote the filtration
generated by $U$ and $\W$. 
\medskip

Consider the stochastic differential equation
\begin{align} 
dX_t &= a(X_t)\,dt + b(X_t; \theta)\, dW_t\,,\quad X_0 = U\,,
\label{eqn:SDE}
\end{align}

for $\theta \in \Theta \subseteq \RR$. The state space of the solution
is assumed to be an open interval $\xx \subseteq \RR$, and the
drift and diffusion coefficients, $a: \xx \to \RR$ and $b: \xx \times
\Theta \to \RR$, are assumed to be known, deterministic functions.
Let $(\PP_\theta)_{\theta \in \Theta}$ be a family of probability
measures on $(\Omega,\ff)$ such that $\X =
(X_t)_{t \geq 0}$ solves (\ref{eqn:SDE}) under
$\PP_\theta$, and let $\EE_\theta$ denote expectation under $\PP_\theta$.
\medskip

Let $t_i^n = i\Delta_n$ with $\Delta_n = 1/n$ for $i\in \NN_0$, $n \in \NN$. For each $n\in \NN$, $\X$ is assumed
to be sampled at times $t_i^n$, $i=0,1,\ldots,n$, yielding the observations $(X_{t_0^n}, X_{t_1^n}, \ldots,
X_{t_n^n})$. Let $\gg_{n,i}$ denote the $\sigma$-algebra generated by
the observations $(X_{t_0^n}, X_{t_1^n}, \ldots,
X_{t_i^n})$, with $\gg_n = \gg_{n,n}$. 

\subsection{Polynomial Growth}\label{sec:poly_asymp}
In the following, to avoid cumbersome notation, $C$ denotes a
generic, strictly positive, real-valued constant. Often, the notation
$C_u$ is used to emphasise that the constant depends on $u$ in some
unspecified manner, where $u$ may be, e.g., a number or a set of
parameter values. Note that, for example, in an expression of the form
$C_u(1+|x|^{C_u})$, the factor $C_u$ and the exponent $C_u$ need not
be equal. Generic constants $C_u$ often depend (implicitly) on the unknown
true parameter value $\theta_0$, but never on the sample size $n$. 
\medskip

A function $f: [0,1] \times \xx^2 \times \Theta \to \RR$ is said to be
of polynomial growth in $x$ and $y$, uniformly for $t\in
  [0,1]$ and $\theta$ in compact, convex sets, if for each compact,
  convex set $K \subseteq \Theta$ there exist constants $C_K>0$ such that
\begin{align*}
\sup_{t \in [0,1],\, \theta \in K} \left|f(t,y,x; \theta)\right| \leq
C_K(1 + |x|^{C_K} + |y|^{C_K})
\end{align*}
for $x,y \in \xx$. 

\begin{definition}
$\cc^{\text{pol}}_{p,q,r}( [0,1]\times \xx^2 \times
  \Theta)$ denotes the class of real-valued functions $f(t,y,x;
  \theta)$ which satisfy that
\begin{enumerate}[label=(\roman{*}), ref=(\roman{*})]
\item\label{asymppol1} $f$ and the mixed partial derivatives
  $\partial^i_t \partial^j_y \partial^k_\theta f(t,y,x; \theta)$,
  $i=0,\ldots,p$, $j=0,\ldots,q$ and $k=0,\ldots,r$ exist and are
  continuous on $[0,1]\times \xx^2 \times
  \Theta$.
\item $f$ and the mixed partial derivatives from \ref{asymppol1} are of polynomial growth in $x$
  and $y$, uniformly for $t\in
  [0,1]$ and $\theta$ in compact, convex sets.
\end{enumerate}
Similarly, the classes $\cc^{\text{pol}}_{p,r}( [0,1] \times
  \xx \times \Theta)$, $\cc^{\text{pol}}_{q,r}( \xx^2
  \times \Theta)$, 
$\cc^{\text{pol}}_{q,r}( \xx \times \Theta)$ and $\cc^{\text{pol}}_{q}( \xx)$
are defined for functions of
the form $f(t,x; \theta)$,  $f(y,x; \theta)$, $f(y;\theta)$ and $f(y)$, respectively.
\defqed
\label{polydef}
\end{definition}
Note that in Definition \ref{polydef}, differentiability of $f$ with
respect to $x$ is never required.
\medskip

For the duration of this paper, $R(t, y,x; \theta)$ denotes a generic, real-valued function
defined on $[0,1] \times
\xx^2 \times \Theta$,
which is of polynomial
growth in $x$ and $y$ uniformly for $t \in [0,1]$ and $\theta$ in compact, convex
sets. The function $R(t, y,x; \theta)$ may depend (implicitly) on
$\thetan$. Functions $R(t,x; \theta)$, $R(y,x; \theta)$ and $R(t,x)$ are defined
correspondingly. The notation $R_\lambda(t,x; \theta)$ indicates that
$R(t,x; \theta)$ also depends on $\lambda \in \Theta$ in an unspecified way.

\subsection{Approximate Martingale Estimating Functions}\label{asymp:amef}

\begin{definition}
Let $g(t,y,x; \theta)$ be a real-valued function defined on $[0,1]
\times \xx^2 \times \Theta$. Suppose the existence of a constant
$\kappa \geq 2$, such that for all $n\in \NN$, $i=1,\ldots,n$, $\theta \in
\Theta$,
\begin{align}
\EE_\theta\left( g(\Delta_n,\xtr,\xtl; \theta) \mid \xtl\right) &=
                                                                  \Delta_n^\kappa
                                                                  R_\theta(\Delta_n,\xtl)\,.
\label{asymp_AMG}
\end{align}
Then, the function 
\begin{align}
G_{n}(\theta) &= \sum_{i=1}^n g(\Delta_n,\xtr,\xtl;
                \theta)
\label{AMEF_asymp}
\end{align}  
is called an \emph{approximate martingale
estimating function}. In particular, when (\ref{asymp_AMG}) is
satisfied with $R_\theta(t,x)\equiv 0$, (\ref{AMEF_asymp}) is referred
to as a \emph{martingale estimating function}. \defqed
\label{def:AMEF_asymp}
\end{definition}
By the Markov property of $\X$, it follows that if $R_\theta(t,x)\equiv 0$,
then $(G_{n,i})_{1\leq i \leq n}$ defined by
\begin{align*}
G_{n,i}(\theta) &= \sum_{j=1}^i g(\Delta_n,X_{t_j^n}, X_{t_{j-1}^n}; \theta)
\end{align*}
 is a zero-mean, real-valued $(\gg_{n,i})_{1\leq i
  \leq n}$-martingale under $\PP_\theta$ for each $n \in \NN$. 
The score function of the observations $(X_{t_0^n}, X_{t_1^n}, \ldots,
X_{t_n^n})$ is a martingale estimating function under weak regularity
conditions, and an approximate martingale estimating function 
can be viewed as an approximation to the score function.  
\medskip

A \emph{$G_n$-estimator} $\hat{\theta}_n$ is essentially obtained as a
solution to the estimating equation $G_n(\theta) = 0$. A more precise
definition is given in the following Definition \ref{def:GnEst_asymp}. 
Here we make the $\omega$-dependence explicit by writing $G_n(\theta,
\omega)$ and $\hat{\theta}_n(\omega)$. 
\begin{definition}
Let $G_n(\theta,\omega)$ be an approximate martingale estimating
function as defined in Definition \ref{def:AMEF_asymp}. Put
$\Theta_\infty = \Theta \cup \{\infty\}$ and let
\begin{align*}
D_n &= \{\omega \in \Omega \mid G_n(\theta, \omega)=0 \text{ has at
  least one solution } \theta \in \Theta\}\,.
\end{align*}

A \emph{$G_n$-estimator} $\hat{\theta}_n(\omega)$ is
any $\gg_n$-measurable function $\Omega \to \Theta_\infty$
which satisfies that for $\PP_\thetan$-almost all $\omega$,
$\hat{\theta}_n(\omega) \in \Theta$ and $G_n(\hat{\theta}_n(\omega),\omega)=0$
if $\omega \in D_n$, while $\hat{\theta}_n(\omega)=\infty$ if $\omega
\notin D_n$.\defqed
\label{def:GnEst_asymp} 
\end{definition}

For any $M_n\neq 0$, the estimating functions $G_n(\theta)$ and
$M_nG_n(\theta)$ yield identical estimators of $\theta$ and are therefore
referred to as \emph{versions} of each other. For any given 
estimating function, it is sufficient that there exists a version of
the function which satisfies the assumptions of this paper, in order
to draw conclusions about the resulting estimators. In particular, we
can multiply by a function of $\Delta_n$.

\subsection{Assumptions}\label{asymp:ass}

We make the following assumptions about the stochastic differential
equation.

\begin{assumption}
The parameter set $\Theta$ is a non-empty, open subset of $\RR$. 
Under the probability measure $\PP_\theta$, the continuous,
$(\ff_t)_{t\geq 0}$-adapted Markov process $\X = (X_t)_{t\geq 0}$
solves a stochastic differential equation of the form (\ref{eqn:SDE}),
the coefficients of which satisfy that 
\begin{align*}
a(y) \in \cc^{\text{pol}}_{6}\left(\xx\right) \quad\text{ and }\quad 
b(y; \theta) \in \cc^{\text{pol}}_{6,2}\left( \xx \times \Theta\right)\,.
\end{align*}
The following holds for all $\theta \in \Theta$.
\begin{enumerate}[label=(\roman{*}), ref=(\roman{*})]
\item 
For all $y\in \xx$, $b^2(y; \theta) > 0$.
\item\label{lips_cont} 
There exists a real-valued constant $C_{\theta}>0$ such that for all
$x,y \in \xx$,
\begin{align*}
\vert
a(x)-a(y)\vert + \vert b(x; \theta) - b(y; \theta)\vert &\leq
C_{\theta}\,\vert x-y\vert\,.
\end{align*}
\item 
$U$ has moments of any order. 
\end{enumerate}\asqed
\label{assumptions_on_X}
\end{assumption}
The global Lipschitz condition, Assumption \ref{assumptions_on_X}.\ref{lips_cont}, 
ensures that a unique solution $\X$ exists. 
The Lipschitz condition and (iii) imply that $\sup_{t \in [0,1]}
\EE_\theta\left(|X_t|^m\right) < \infty$ for all $m\in \NN$.
Assumption \ref{assumptions_on_X} is very similar to the corresponding
Condition $2.1$ of \citet{efficient}. However, an important difference
is that in the current paper, $\X$ is not required to be
ergodic. Here, law of large numbers-type results are proved by what is, in
essence, the convergence of Riemann sums.
\medskip

We make the following assumptions about the estimating function.

\begin{assumption}
The function $g(t,y,x; \theta)$ satisfies (\ref{asymp_AMG}) for some
$\kappa \geq 2$, thus defining an approximate martingale
estimating function by (\ref{AMEF_asymp}). Moreover,
\begin{align*}
g(t,y,x; \theta)& \in \cc^{\text{pol}}_{3,8,2}( [0,1] \times \xx^2
  \times \Theta)\,,
\end{align*}
and the following holds for all $\theta \in \Theta$.
\begin{enumerate}[label=(\roman{*}), ref=(\roman{*}), resume]
\item \label{dyg0_cont}For all $x\in \xx$, $\partial_y g(0,x,x;
\theta) =0$.
\item 
The expansion 
\begin{align}
\begin{split}
g(\Delta, y, x; \theta) &= g(0,y,x; \theta) + \Delta g^{(1)}(y,x;
\theta) + \tfrac{1}{2}\Delta^2 g^{(2)}(y,x; \theta) + \tfrac{1}{6}
\Delta^3 g^{(3)}(y,x; \theta) \\
&\hspace{5mm} + \Delta^4 R(\Delta,y ,x; \theta)
\label{eqn:g_taylor}
\end{split}
\end{align}
holds for all $\Delta \in [0,1]$ and $x,y \in \xx$, where
$g^{(j)}(y,x; \theta)$ denotes the $j'$th partial derivative of
$g(t,y,x; \theta)$ with respect to $t$, evaluated in $t=0$.

\asqed
\end{enumerate}

\label{assumptions_on_g}
\end{assumption}

Assumption \ref{assumptions_on_g}.\ref{dyg0_cont} was referred to by
\citet{efficient} as \emph{Jacobsen's condition}, as it is one of the
conditions for small $\Delta$-optimality in the sense of
\citet{jacobsen2001}, see \citet{jacobsen2002}. The assumption ensures
rate optimality of the estimators in this paper, and of the estimators of
the parameters in the diffusion coefficient in \citet{efficient}.
\medskip

The assumptions of polynomial growth and existence and boundedness of
all moments serve to simplify the exposition and proofs, and could be
relaxed. 

\subsection{The Infinitesimal Generator}\label{asymp:infini}

For $\lambda \in \Theta$, the infinitesimal generator $\ll_\lambda$ is
defined for all functions $f(y) \in \cc_2^{\text{pol}}(\xx)$ by
\begin{align*}
\ll_{\lambda}f(y) &= a(y)\partial_y f(y) + \tfrac{1}{2}
b^2(y;\lambda) \partial^2_y f(y)\,.
\end{align*}
For $f(t, y,x, \theta) \in \cc^{\text{pol}}_{0,2,0,0}(
[0,1] \times \xx^2 \times \Theta)$, let
\begin{align}
\ll_{\lambda}f(t,y,x; \theta) &= a(y)\partial_y f(t,y,x;\theta) + \tfrac{1}{2}
b^2(y;\lambda) \partial^2_y f(t, y, x; \theta)\,.
\label{L4var}
\end{align}
Often, the notation $\ll_{\lambda}f(t,y,x; \theta) = \ll_\lambda (f(t;
\theta))(y,x)$ is used, so e.g. $\ll_\lambda (f(0;
\theta))(x,x)$ means $\ll_{\lambda}f(0,y,x; \theta)$ evaluated in $y=x$. 
In this paper the infinitesimal generator is particularly useful
because of the following result. 

\begin{lemma}
Suppose that Assumption \ref{assumptions_on_X} holds, and that for
some $k\in \NN_0$,
\begin{align*}
a(y) \in \cc^{\text{pol}}_{2k}(\xx)\,, \quad b(y; \theta)
\in \cc^{\text{pol}}_{2k,0}(\xx \times \Theta)\quad \text{ and }\quad f(y,x; \theta) \in \mathcal{C}^{\text{pol}}_{2(k+1),0}(
  \xx^2 \times \Theta)\,.
\end{align*}
Then, for $0 \leq t \leq t+ \Delta\leq 1$ and
$\lambda \in \Theta$,
\begin{align*}
\begin{split}
& \EE_{\lambda} \left( f(X_{t+\Delta}, X_t; \theta) \mid X_t\right) \nonumber\\
&= \sum_{i=0}^k \frac{\Delta^i}{i!} \ll_{\lambda}^i f(X_t, X_t; \theta) + \int_0^\Delta
\int_0^{u_1} \cdots \int_0^{u_k} \EE_{\lambda} \left(
  \ll_{\lambda}^{k+1}f(X_{t+u_{k+1}}, X_t; \theta)\mid X_t\right) 
\, d u_{k+1} \cdots du_1
\end{split}
\end{align*}
where, furthermore, 
\begin{align*}
\int_0^\Delta\int_0^{u_1} \cdots \int_0^{u_k} \EE_{\lambda} \left(
  \ll_{\lambda} ^{k+1}f(X_{t+u_{k+1}}, X_t; \theta)\mid X_t\right) 
\, d u_{k+1} \cdots du_1 &= \Delta^{k+1}R_\lambda(\Delta,X_t; \theta)\,.
\end{align*}\theoqed
\label{stoch_taylor}
\end{lemma}
The expansion of
the conditional expectation in powers of $\Delta$ in the first part of
the lemma corresponds to Lemma 1 in \citet{fz1989} and Lemma 4 in
\citet{dc_fz}. It may be proven by induction on $k$ using It\^{o}'s
formula, see, e.g., the proof of \citet[Lemma 1.10]{MSbog}.
The characterisation of the remainder term follows by applying
Corollary \ref{f_growth} to $\ll_{\lambda}^{k+1}f$, see the proof of
\citet[Lemma $1$]{kessler_ergodic}. 
\medskip

For concrete models, Lemma \ref{stoch_taylor} is useful for
verifying the approximate martingale property (\ref{asymp_AMG}) and
for creating approximate martingale estimating functions. In
combination with (\ref{asymp_AMG}), the lemma is key to proving the
following Lemma \ref{lemma_properties}, which reveals two important
properties of approximate martingale estimating functions. 
\begin{lemma}
Suppose that Assumptions \ref{assumptions_on_X}
and \ref{assumptions_on_g} hold. Then 
\begin{align*}
g(0,x,x; \theta) = 0\quad\text{ and }\quad g^{(1)}(x,x; \theta) &=
-\ll_\theta(g(0,\theta))(x,x)
\end{align*}

for all $x \in \xx$ and $\theta \in \Theta$.\theoqed
\label{lemma_properties}
\end{lemma}
Lemma \ref{lemma_properties} corresponds to Lemma 2.3 of \citet{efficient},
to which we refer for details on the proof.

\section{Main Results}\label{sec:asymp:main}
Section \ref{asymp:maintheo} presents the main theorem of this paper,
which establishes existence, uniqueness and asymptotic distribution
results for rate optimal estimators based on approximate martingale estimating
functions. In Section \ref{sec:asymp:eff} a condition is given, which
ensures that the rate optimal estimators are also efficient, and
efficient estimators are discussed.

\subsection{Main Theorem}\label{asymp:maintheo}
The final assumption needed for the main theorem is as follows.
\begin{assumption}
The following holds $\PP_\theta$-almost surely for all $\theta \in \Theta$.
\begin{enumerate}[label=(\roman{*}), ref = (\roman{*})]
\item \label{item:anot0}
For all $\lambda\neq \theta$,
\begin{align*}
\int_0^1 \big( b^2(X_s;
  \theta) - b^2(X_s; \lambda)\big) \partial^2_y g(0,X_s,X_s;
  \lambda)\, ds \neq 0\,,
\end{align*}
\item \label{item:bnot0}
\begin{align*}
\int_0^1 \partial_\theta b^2(X_s;\theta) \partial_y^2
  g(0,X_s, X_s;\theta)\, ds \neq 0\,,
\end{align*}
\item \label{item:cnot0}
\begin{align*}
\int_0^1 b^4(X_s; \theta) \left(\partial_y^2 g(0,X_s,X_s;
  \theta)\right)^2\, ds \neq 0\,.
\end{align*}\asqed
\end{enumerate}
\label{assumptions_for_maintheo}
\end{assumption}
Assumption \ref{assumptions_for_maintheo} can be
difficult to check in practice because it involves the full sample path of $\X$ over
the interval $[0,1]$. It requires, in particular, that for all $\theta \in \Theta$, with
$\PP_\theta$-probability one, $t\mapsto b^2(X_t; \theta)-b^2(X_t;
\lambda)$ is not Lebesgue-almost surely zero when $\lambda \neq
\theta$. As noted by \citet{gc_jacod}, this requirement holds
true (by the continuity of the function) if, for example, $X_0 = U$ is
degenerate at $x_0$, and $b^2(x_0; \theta) \neq b^2(x_0; \lambda)$ for
all $\theta \neq \lambda$. 
\medskip

For an efficient estimating function, Assumption
\ref{assumptions_for_maintheo} reduces to conditions on $\X$ with no
further conditions on the estimating function, see the next
section. Specifically, the conditions involve only the squared
diffusion coefficient $b^2(x;\theta)$ and its derivative
$\partial_\theta b^2$. 
\medskip

\begin{theorem}
Suppose that Assumptions \ref{assumptions_on_X},
\ref{assumptions_on_g} and \ref{assumptions_for_maintheo} hold. Then,
\begin{enumerate}[label=(\roman{*}), ref=(\roman{*})]
\item\label{item:existence_asymp}
 there exists a consistent $G_n$-estimator
  $\hat{\theta}_n$. Choose any compact, convex set $K\subseteq \Theta$
  with $\theta_0 \in \text{int}\, K$, where $\text{int}\, K$ denotes
  the interior of $K$. Then, the consistent $G_n$-estimator $\hat{\theta}_n$ is eventually
  unique in $K$, in the sense that for any $G_n$-estimator $\tilde{\theta}_n$ with 
 $\PP_{\theta_0}(\tilde{\theta}_n \in K) \to 1$ as $n\to \infty$, it
 holds that
 $\PP_{\theta_0}(\hat{\theta}_n\neq\tilde{\theta}_n) \to 0$
 as $n\to \infty$.
\item\label{item:main_asymp}
for any consistent $G_n$-estimator $\hat{\theta}_n$, it holds that
\begin{align}
\sqrt{n}( \hat{\theta}_n - \theta_0)
  \overset{\dd}{\longrightarrow} W(\theta_0)Z\,.
\label{rate_limit}
\end{align}
The limit distribution is a normal variance-mixture, where $Z$ is
standard normal distributed, and independent of 
$W(\theta_0)$ given by
\begin{align}
W(\theta_0) &= \frac{\left(\displaystyle\int_0^1\tfrac{1}{2} b^4(X_s;
\theta_0) \left(\partial^2_y g(0,X_s,X_s; \theta_0)\right)^2 \,
ds\right)^{1/2}}{
  \displaystyle\int_0^1 \tfrac{1}{2} \partial_\theta b^2(X_s;\theta_0) \partial_y^2
  g(0,X_s, X_s;\theta_0)\, ds}\,. 
\label{W}
\end{align}
\item \label{item:asymp_estvar}
for any consistent $G_n$-estimator $\hat{\theta}_n$, 
\begin{align}
\widehat{W}_n &= - \frac{\left( \displaystyle \frac{1}{\Delta_n} \sum_{i=1}^n
    g^2(\Delta_n, \xtr, \xtl;
    \hat{\theta}_n)\right)^{1/2}}{\displaystyle
  \sum_{i=1}^n \partial_\theta g(\Delta_n, \xtr, \xtl;
  \hat{\theta}_n)}
\label{What}
\end{align}
satisfies that $\widehat{W}_n \overset{\pp}{\longrightarrow} W(\theta_0)$, and
\begin{align*}
\sqrt{n} \,\widehat{W}_n^{-1} ( \hat{\theta}_n - \theta_0)
\overset{\dd}{\longrightarrow} \nn (0,1)\,.
\end{align*}
\end{enumerate}
\label{est_exist}
\theoqed\end{theorem}
The proof of Theorem \ref{est_exist} is given
in Section \ref{sec:main_lemma}.
\medskip

\citet{dohnal1987} and \citet{gobet2001} showed local asymptotic mixed
normality with rate $\sqrt{n}$, so Theorem  \ref{est_exist}
establishes rate optimality of $G_n$-estimators.
\medskip

Observe that the limit distribution in Theorem
\ref{est_exist}.\ref{item:main_asymp} generally depends on not only the unknown
parameter $\thetan$, but also on the concrete realisation of the sample path $t \mapsto X_t$ over
$[0,1]$, which is only partially observed. Note also that a
variance-mixture of normal distributions can be very different from a
Gaussian distribution. It can be much more heavy-tailed and even have
no moments. Theorem \ref{est_exist}.\ref{item:asymp_estvar} is
therefore important as it
yields a standard normal limit distribution, which is more useful in
practical applications. 

\subsection{Efficiency}
\label{sec:asymp:eff}
Under the assumptions of Theorem \ref{est_exist}, the following additional
condition ensures efficiency of a consistent $G_n$-estimator.
\begin{assumption}
Suppose that for each $\theta \in \Theta$, there exists a constant $K_\theta \neq 0$ such that for all $x\in \xx$,
\begin{align*}
\partial^2_y g(0,x,x; \theta) = K_\theta \frac{\partial_\theta
    b^2(x; \theta)}{b^4(x; \theta)}\,.
\end{align*}\asqed
\label{efficiency}
\end{assumption}

\citet{dohnal1987} and \citet{gobet2001} showed
that the local asymptotic mixed normality property holds
within the framework considered here with random Fisher information
$\mathcal{I}(\theta_0)$ given by (\ref{I}). Thus, a $G_n$-estimator
$\hat{\theta}_n$ is efficient if (\ref{rate_limit}) holds with
$W(\thetan) = \mathcal{I}(\theta_0)^{-1/2}$, and the following
Corollary \ref{cor_effi} may easily be verified.
\begin{corollary}
Suppose that the assumptions of Theorem \ref{est_exist} and Assumption
\ref{efficiency} hold. Then, any consistent $G_n$-estimator is also
efficient.\theoqed
\label{cor_effi}
\end{corollary}
It follows from Theorem \ref{est_exist}  and Lemma \ref{main_lemma_P}
that if Assumption \ref{efficiency} holds, and if $G_n$ is normalized
such that $K_\theta=1$, then
\begin{align*}
\sqrt{n} \,\widehat{\cal I}_n^{\frac12} ( \hat{\theta}_n - \theta_0)
\overset{\dd}{\longrightarrow} \nn (0,1)\,,
\end{align*}
where
\[
\widehat{\cal I}_n = \frac{1}{\Delta_n} \sum_{i=1}^n
 g^2(\Delta_n, \xtr, \xtl; \hat{\theta}_n).
\]
\medskip

It was noted in Section \ref{asymp:amef} that not necessarily all
versions of a particular estimating function satisfy the conditions of 
this paper, even though they lead to the same estimator. Thus, an
estimating function is said to be efficient, if there exists a version
which satisfies the conditions of Corollary \ref{cor_effi}. The same
goes for rate optimality. 
\medskip

Assumption \ref{efficiency} is identical to the condition for
efficiency of estimators of parameters in the diffusion coefficient in
\citet{efficient}, and to one of the conditions for
small $\Delta$-optimality given in \citet{jacobsen2002}. 
\medskip

Under suitable regularity conditions on the diffusion coefficient $b$, the function 
\begin{align}
\bar{g}(t,y,x; \theta) &= \frac{\partial_\theta b^2(x;
  \theta)}{b^4(x; \theta)} \left( (y-x)^2 - t b^2(x; \theta)\right)
\label{effi_g}
\end{align}
yields an example of an efficient estimating function. The approximate
martingale property (\ref{asymp_AMG}) can be verified by Lemma \ref{stoch_taylor}.
\medskip

When adapted to the current framework, the contrast functions
investigated by \citet{gc_jacod} have the form
\begin{align*}
U_n(\theta) &= \frac{1}{n} \sum_{i=1}^n f\left( b^2(\xtl; \theta), \Delta_n^{-1/2}(\xtr-\xtl)\right)\,,
\end{align*}
for functions $f(v,w)$ satisfying certain conditions. For the contrast
function identified as efficient by
\citeauthor{gc_jacod}, $f(v,w) = \log v + w^2/v$. Using that $\Delta_n
= 1/n$, it is then seen that their efficient contrast function is of
the form $\bar{U}_n(\theta) =
\sum_{i=1}^n \bar{u}(\Delta_n,\xtr,\xtl; \theta)$ with 
\begin{align*}
\bar{u}(t,y,x; \theta) &= t\log b^2(x; \theta) + (y-x)^2/b^2(x; \theta)
\end{align*}
and $\partial_\theta
\bar{u}(t,y,x; \theta) = -\bar{g}(t,y,x; \theta)$. In other words, it corresponds
to a version of the efficient approximate martingale estimating function given
by (\ref{effi_g}). The same contrast function was considered by
\citet{uchidayoshida} in the framework of a more general class of
stochastic differential equations.
\medskip

A problem of considerable practical interest is how to construct
estimating functions that are rate optimal and efficient,
i.e. estimating functions satisfying Assumptions
\ref{assumptions_on_g}.\ref{dyg0_cont} and \ref{efficiency}. Being the
same as the conditions for small $\Delta$-optimality, the assumptions
are, for example, satisfied for martingale estimating functions
constructed by \citet{jacobsen2002}.
\medskip

As discussed by \citet{efficient}, the rate optimality and efficiency
conditions are also satisfied by Godambe-Heyde optimal approximate
martingale estimating functions. Consider martingale estimating functions of the form
\begin{align*}
g(t,y,x;\theta)
&= a(x,t;\theta)^* \left(f(y;\theta) - \phi^t_\theta f(x;\theta)\right)\,,
\end{align*}
where $a$ and $f$ are two-dimensional, $*$ denotes transposition, and 
$\phi^t_\theta f(x;\theta)   = \EE_\theta (f(X_t; \theta) \mid
X_0=x)$. Suppose that $f$ satisfies appropriate (weak) 
conditions. Let $\bar{a}$ be the
weight function for which the estimating function is optimal in the
sense of \citet{gh1987}, see e.g. \citet{heyde1997} or
\citet[Section 1.11]{MSbog}. It follows by an argument analogous to
the proof of Theorem $4.5$ in \citet{efficient} that the estimating
function with
\begin{align*}
g(t,y,x;\theta) = t \bar{a}(x,t;\theta)^*[f(y;\theta) -
  \phi^t_\theta f(x;\theta)]
\end{align*}
satisfies Assumptions
\ref{assumptions_on_g}.\ref{dyg0_cont} and \ref{efficiency}, and is
thus rate optimal and efficient. As there is a simple formula for
$\bar{a}$ (see Section 1.11.1 of \citet{MSbog}), this provides a way of
constructing a large number of efficient estimating functions. The result also
holds if $\phi^t_\theta f(x;\theta)$ and the conditional moments in the
formula for $\bar{a}$ are suitably approximated by the help of Lemma
\ref{stoch_taylor}.  
\medskip

\begin{remark}
Suppose for a moment that the diffusion coefficient of (\ref{eqn:SDE})
has the form $b^2(x; \theta) = h(x)k(\theta)$ for strictly
positive functions $h$ and $k$, with Assumption \ref{assumptions_on_X}
satisfied. This holds true, e.g., for a number of Pearson diffusions,
including the (stationary) Ornstein-Uhlenbeck 
and square root processes. (See \citet{forman2008} for more on Pearson
diffusions.) Then $\ii(\theta_0)
= \partial_\theta k(\theta_0)^2/(2k^2(\theta_0))$. In this case, under the assumptions
of  Corollary \ref{cor_effi}, an efficient $G_n$-estimator
$\hat{\theta}_n$ satisfies that $\sqrt{n} (\hat{\theta}_n - \theta_0)
\longrightarrow Y$ in distribution where $Y$ is normal distributed
with mean zero and variance $2k^2(\theta_0)/\partial_\theta
k(\theta_0)^2$, i.e.\ the limit distribution is not a normal
variance-mixture depending on $(X_t)_{t\in [0,1]}$. 
Note also that when $b^2(x; \theta) = h(x)k(\theta)$ and Assumption
\ref{efficiency} holds, then Assumption
\ref{assumptions_for_maintheo} is satisfied when, e.g., $\partial_\theta
k(\theta)>0$ or $\partial_\theta k(\theta)<0$.
\reqed\label{remark:pearson}
\end{remark}
\medskip

\section{Simulation study}\label{sec:asymp:sim}

This section presents a simulation study illustrating the theory in the
previous section. An efficient and an inefficient estimating
function are compared for two models, an ergodic and a non-ergodic
model. For both models the limit distributions of the consistent
estimators are non-degenerate normal variance-mixtures. 
\medskip

First, consider the stochastic differential equation
\begin{align}
dX_t &= -2 X_t \, dt + (\theta+X_t^2)^{-1/2}\, dW_t, \ \ \ X_0=0,
\label{SDEex}
\end{align}
where $\theta \in (0,\infty)$ is an unknown parameter. The solution
$\X$ is ergodic with invariant probability density proportional to
$\exp\left( -2 \theta x^2 - x^4\right)\left(\theta + x^2\right)$, $x \in \RR$.
The process satisfies Assumption \ref{assumptions_on_X}. We compare
the two estimating functions given by 
\begin{align*}
G_n(\theta) = \sum_{i=1}^n g(\Delta_n, \xtr, \xtl; \theta)\quad\text{
  and }\quad H_n(\theta) = \sum_{i=1}^n h(\Delta_n, \xtr, \xtl; \theta)
\end{align*}
where
\begin{align*}
g(t,y,x; \theta) &=  (y-(1-2t)x)^2-(\theta+x^2)^{-1}t  \\
h(t,y,x; \theta) &= (\theta+x^2)^{10} (y-(1-2t)x)^2-(\theta+x^2)^9t \,.
\end{align*} 
Both $g$ and $h$ satisfy Assumptions \ref{assumptions_on_g} and
\ref{assumptions_for_maintheo}. Moreover, $g$ satisfies the condition
for efficiency, while $h$ is not efficient. 
\medskip

Let $W_G(\thetan)$ and $W_H(\thetan)$ be given by
(\ref{W}), that is
\begin{align}
W_G(\theta_0) =
- \left(\tfrac{1}{2}\int_0^1 \frac{1}{(\theta_0+X_s^2)^2}\,
  ds\right)^{-1/2}\hspace{1mm}\text{and}\hspace{2mm} W_H(\theta_0) = -\frac{\left(
  \displaystyle \int_0^1 2 (\theta_0+X_s^2)^{18}\,
  ds\right)^{1/2}}{\displaystyle \int_0^1 (\theta_0+X_s^2)^8\, ds}\,. 
\label{WGH}
\end{align}
Numerical calculations and simulations were done in R
3.1.3 \citep{R}. First, $m=10^4$ trajectories of the process $\X$ given by
(\ref{SDEex}) were simulated over the time-interval $[0,1]$ with 
$\theta_0 = 1$. These simulations were performed using the Milstein
scheme as implemented in the R-package \emph{sde} \citep{sde} with
step size $10^{-5}$. The simulations were subsampled to obtain samples
sizes of $n=10^3$ and $n=10^4$. Let $\hat{\theta}_{G,n}$ and 
$\hat{\theta}_{H,n}$ denote estimates of $\theta$
obtained by solving the equations $G_n(\theta)=0$ and
$H_n(\theta)=0$ numerically, on the interval $[0.01, 1,99]$. Using these estimates,
$\widehat{W}_{G,n}$ and $\widehat{W}_{H,n}$ were calculated by
(\ref{What}). For $n=10^3$, $\hat{\theta}_{H,n}$ could not be computed
for $30$ of the $m=10^4$ sample paths. For $n=10^4$, and for the
efficient estimator $\hat{\theta}_{G,n}$ there were no problems.  
\medskip

Figure \ref{fig:qq} shows QQ-plots of 
\begin{align*}
\widehat{Z}_{G,n} = \sqrt{n}\,\widehat{W}_{G,n}^{-1}(\hat{\theta}_{G,n}-\theta_0)\quad\text{ and
}\quad \widehat{Z}_{H,n} = \sqrt{n}\,\widehat{W}_{H,n}^{-1}(\hat{\theta}_{H,n}-\theta_0)\,,
\end{align*}
compared with a standard normal distribution, for $n=10^3$ and
$n=10^4$ respectively.
\begin{figure}
\begin{center}
\includegraphics[width=0.45\textwidth]{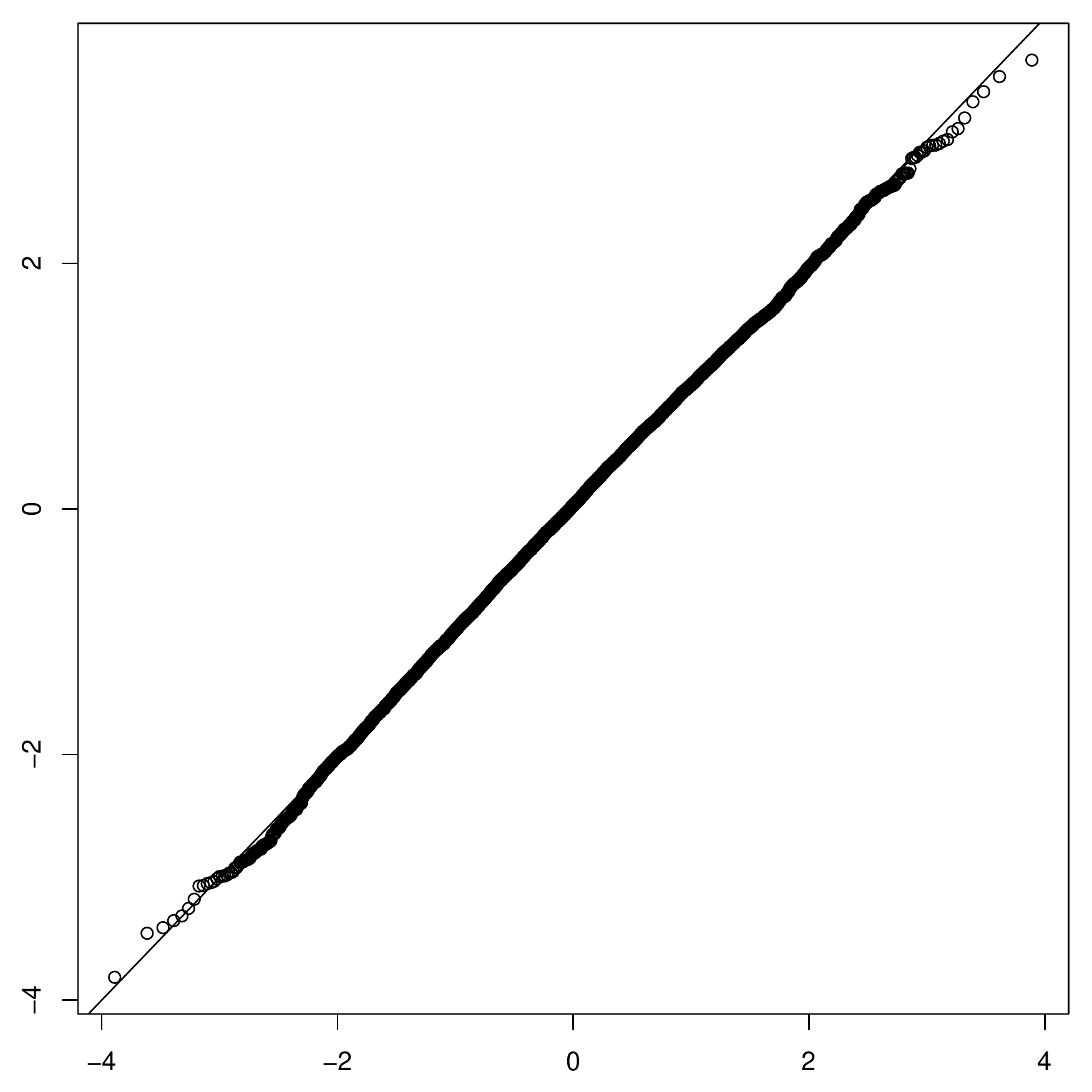}
\includegraphics[width=0.45\textwidth]{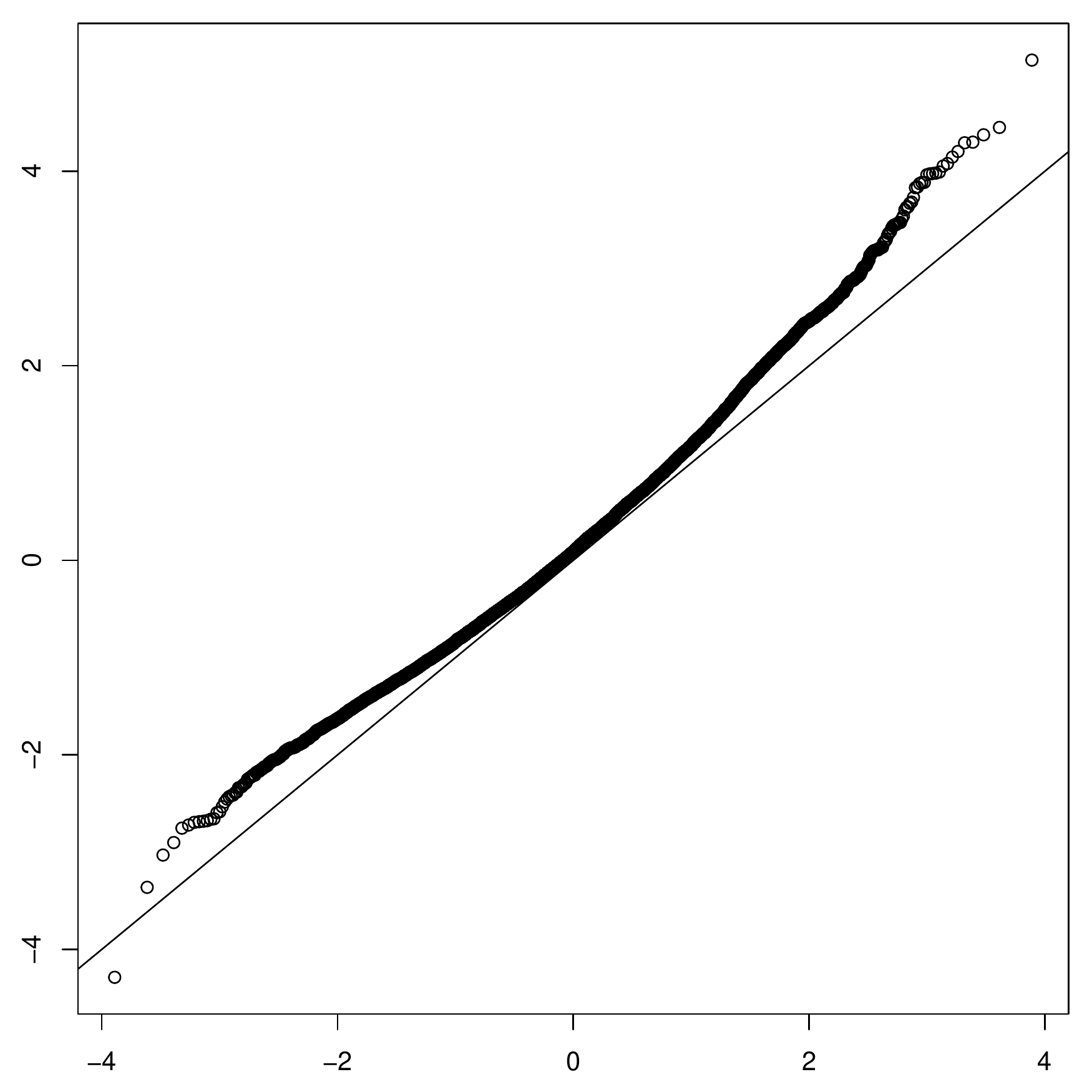}
\includegraphics[width=0.45\textwidth]{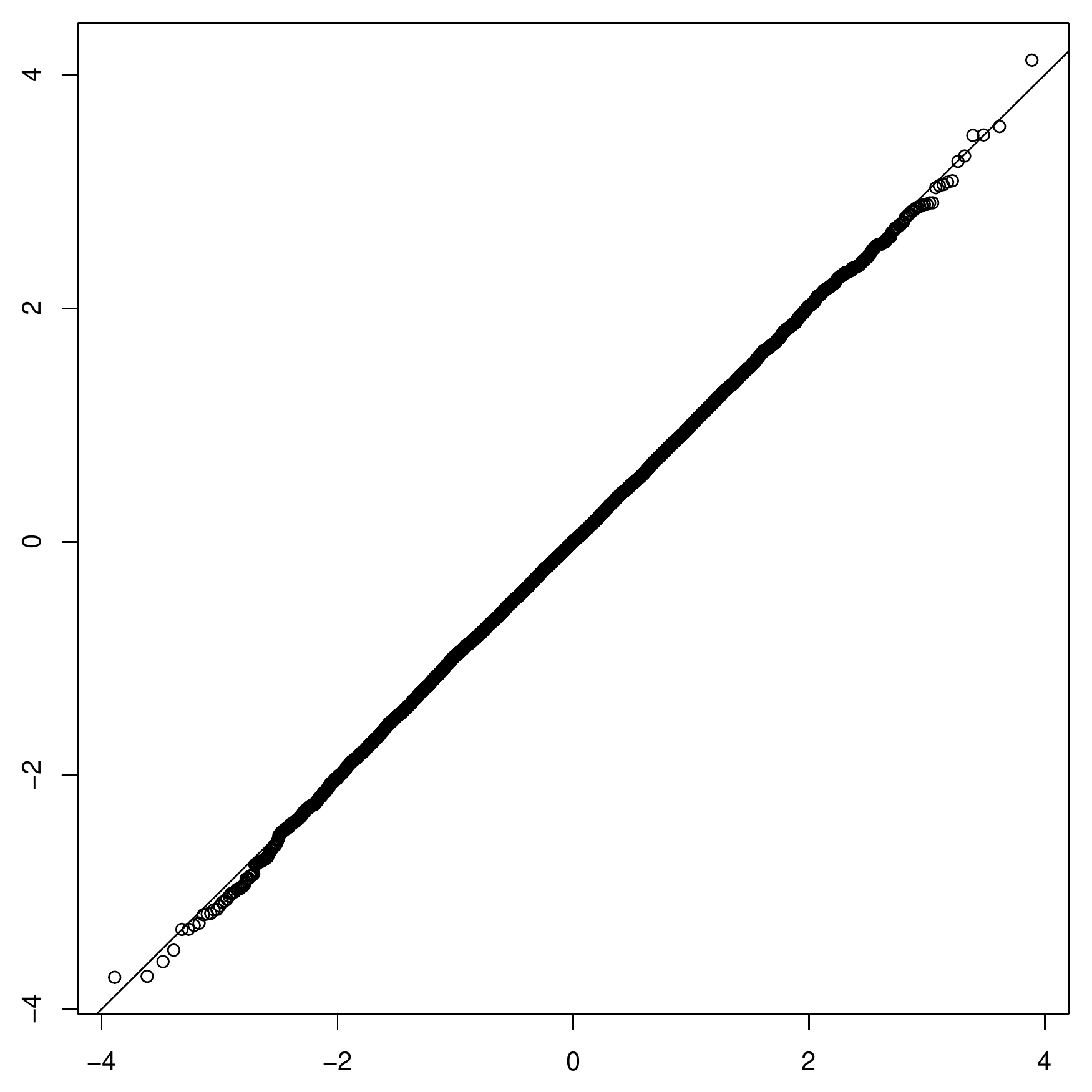}
\includegraphics[width=0.45\textwidth]{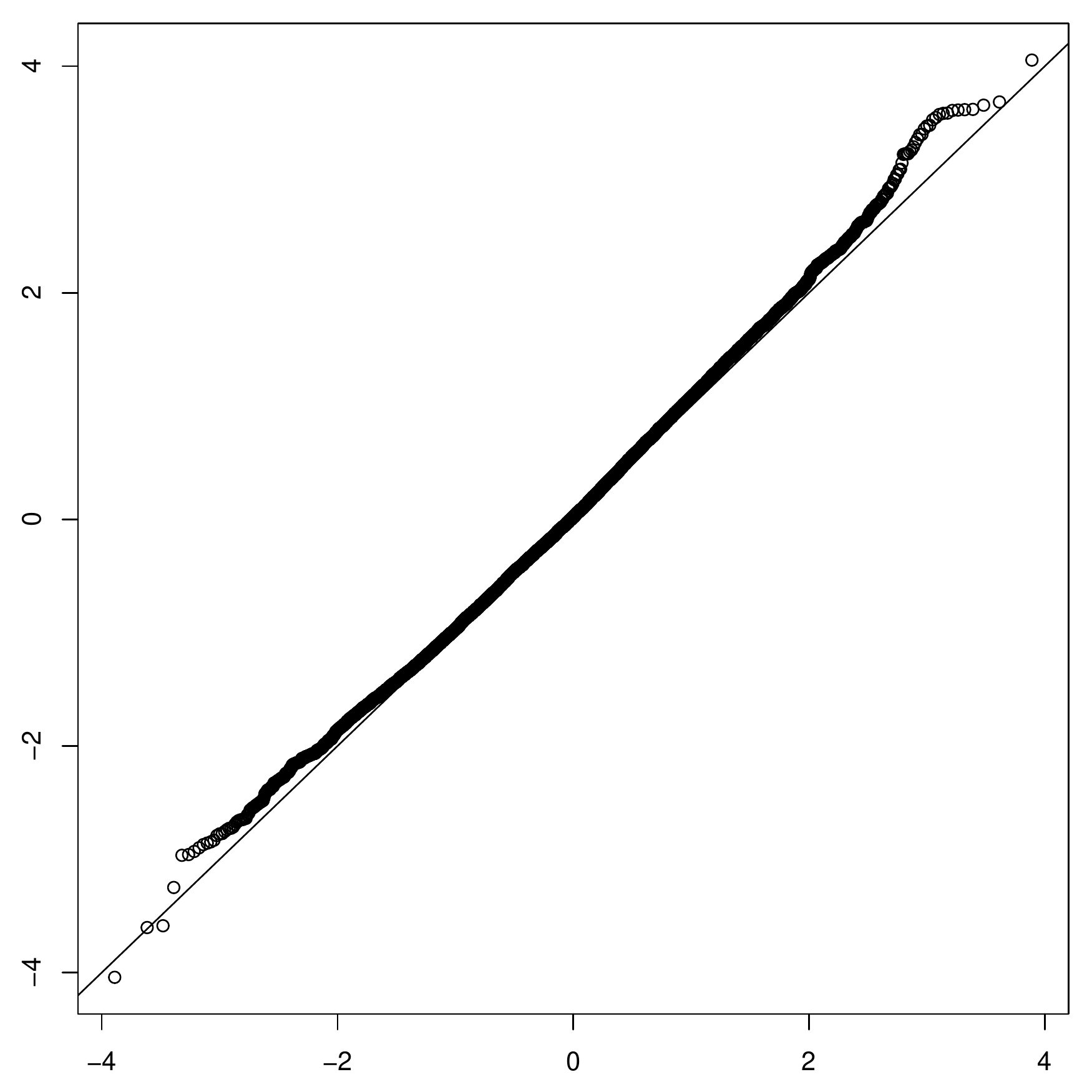}
\end{center}
\caption{QQ-plots comparing $\widehat{Z}_{G,n}$ (left) and
$\widehat{Z}_{H,n}$ (right) to the $\nn(0,1)$ distribution in the case
of the ergodic model (\ref{SDEex}) for
$n=10^3$ (above) and $n=10^4$ (below).} 
\label{fig:qq}
\end{figure} 
These QQ-plots suggest that, as
$n$ goes to infinity, the asymptotic
distribution in Theorem
\ref{est_exist}.\ref{item:asymp_estvar} becomes a good approximation
faster in the efficient case than in the inefficient
case.
\medskip 

Inserting $\thetan=1$ into (\ref{WGH}), the integrals in these
expressions may be approximated by Riemann sums, using each of the 
simulated trajectories of $\X$ (with sample size $n=10^4$ for maximal
accuracy). This method yields a second set of approximations
$\widetilde{W}_{G}$ and $\widetilde{W}_{H}$ to the realisations of the
random variables 
$W_G(\theta_0)$ and $W_H(\theta_0)$, presumed to be more accurate than
$\widehat{W}_{G,10^4}$ and $\widehat{W}_{H,10^4}$ as they utilise the
true parameter value. The \emph{density} function in R was used (with
default arguments) to compute an approximation to 
the densities of $W_G(\theta_0)$ and $W_H(\theta_0)$, using the
approximate realisations $\widetilde{W}_{G}$ and
$\widetilde{W}_{H}$. 
\begin{figure}
\begin{center}
\includegraphics[width=0.45\textwidth]{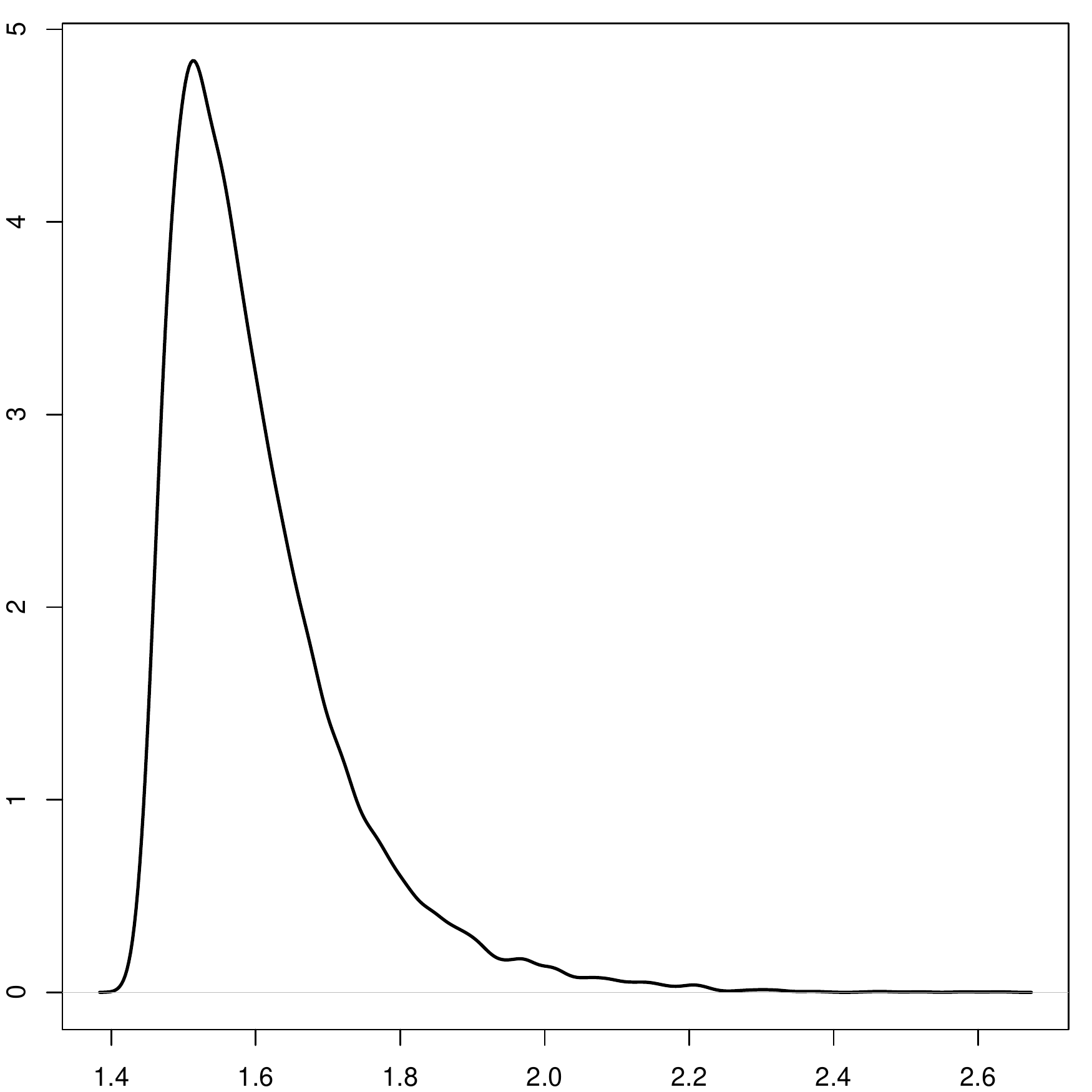}
\includegraphics[width=0.45\textwidth]{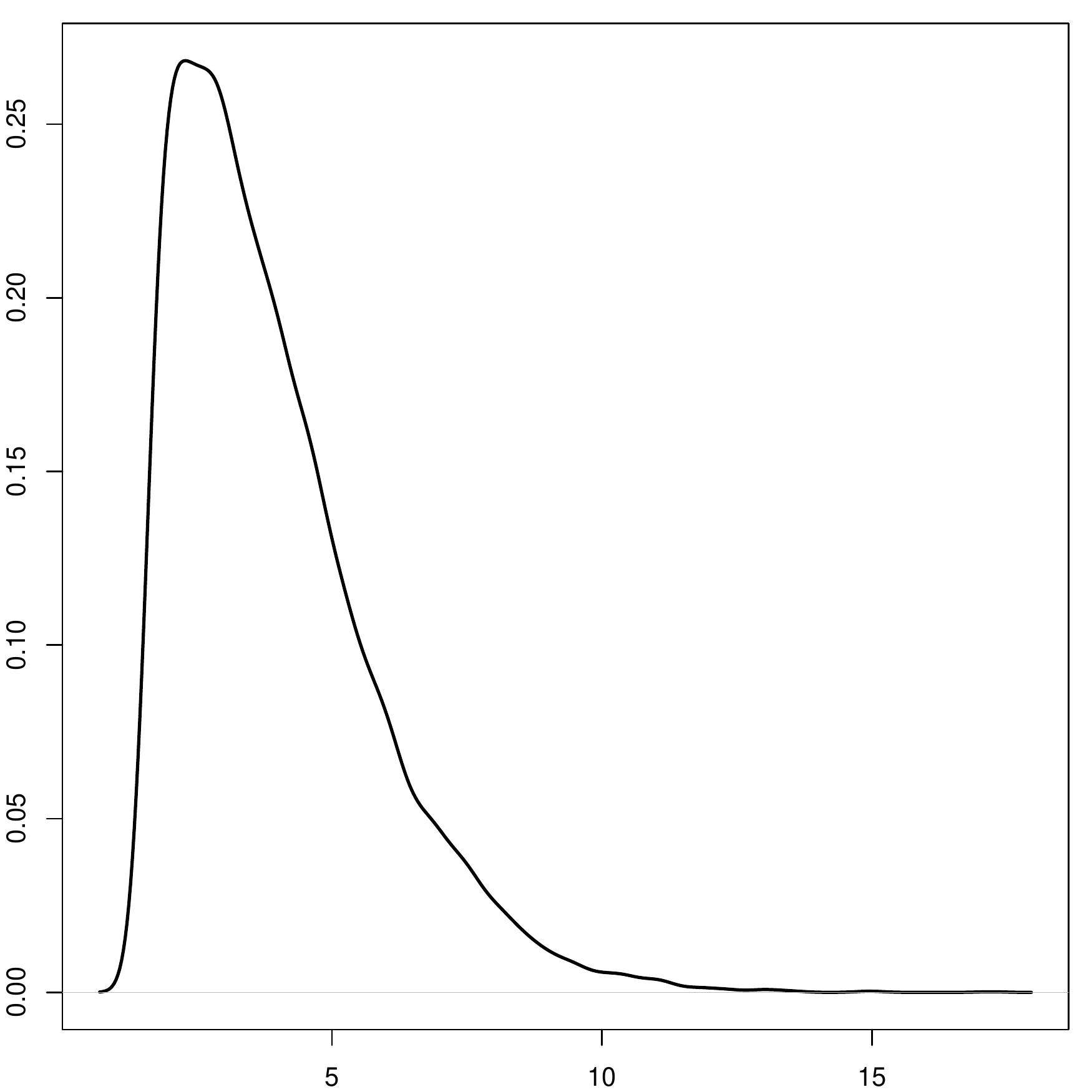}
\end{center}
\caption{Approximation to the densities of $W_G(\theta_0)$ (left) and
$W_H(\theta_0)$ (right) based on $\widetilde{W}_{G}$ and
$\widetilde{W}_{H}$ in the case of the ergodic model (\ref{SDEex}).}
\label{fig:dense}
\end{figure}   
\medskip

It is seen from Figure \ref{fig:dense} 
that the distribution of
$W_H(\theta_0)$ is much more spread out
than the distribution of $W_G(\theta_0)$. This corresponds well to the
limit distribution in Theorem
\ref{est_exist}.\ref{item:main_asymp} being more spread out in the
inefficient case than in the efficient case. 
\begin{figure}
\begin{center}
\includegraphics[width=0.45\textwidth]{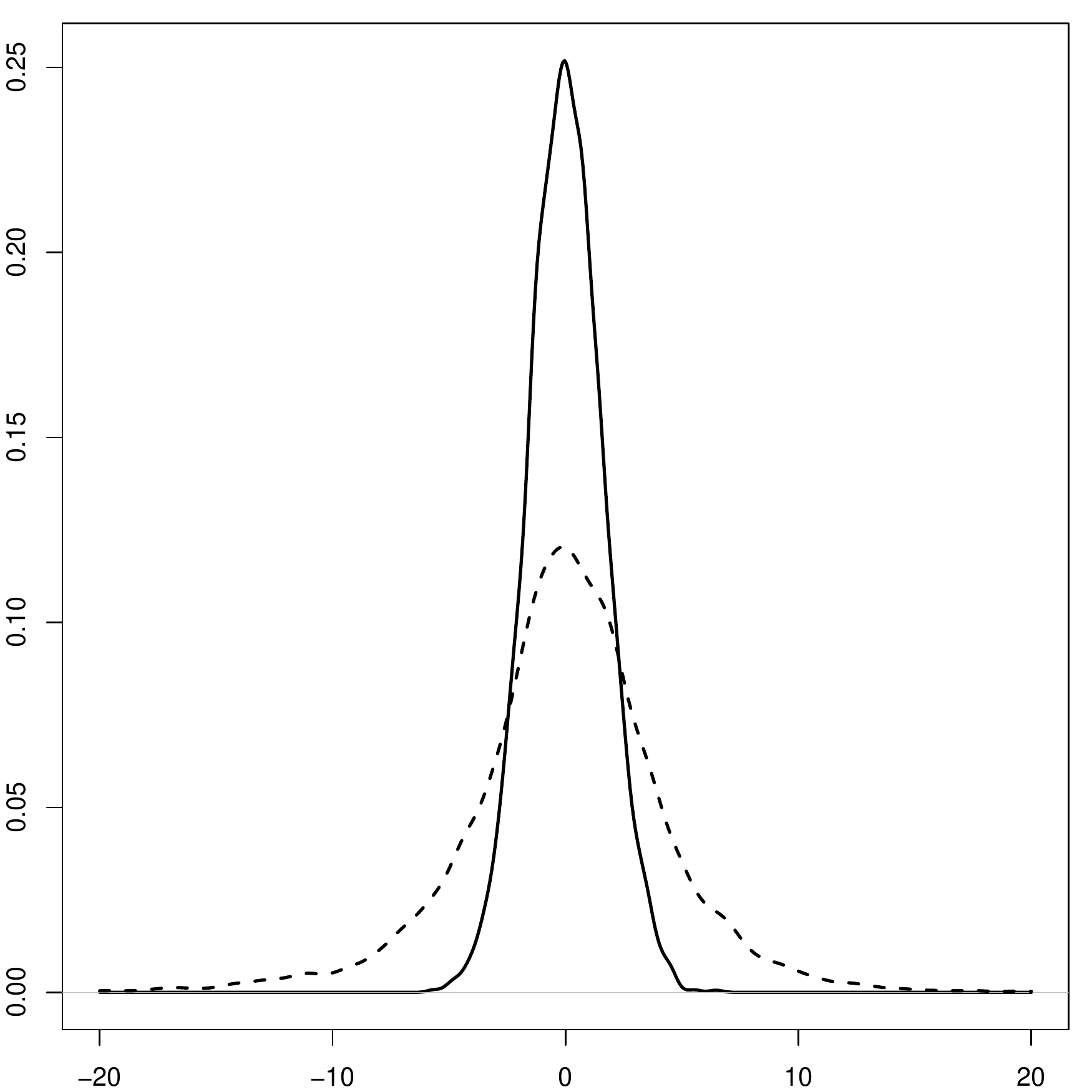}
\end{center}
\caption{Estimated densities of $\sqrt{n}(\hat{\theta}_{G,n}-\theta_0)$ (solid curve) and
  $\sqrt{n}(\hat{\theta}_{H,n}-\theta_0) $ (dashed curve) for $n=10^4$
  in the case of the ergodic model (\ref{SDEex}).} 
\label{fig:visual}
\end{figure}  
Along the same lines, Figure \ref{fig:visual} 
shows similarly
computed densities based on $\sqrt{n}(\hat{\theta}_{G,n}-\theta_0)$ and
$\sqrt{n}(\hat{\theta}_{H,n}-\theta_0)$ for $n=10^4$,
which may be considered approximations to the densities of the normal
variance-mixture limit distributions in Theorem
\ref{est_exist}.\ref{item:main_asymp}. These plots also illustrate
that the limit distribution of the inefficient estimator is more
spread out than that of the efficient estimator. 
\medskip 

Now, consider the stochastic differential equation
\begin{align}
dX_t &= 2 X_t \, dt + (\theta+X_t^2)^{-1/2}\, dW_t, \ \ \ X_0=0.
\label{SDEex2}
\end{align}
For this model,
the solution $\X$ is not ergodic, but again Assumption
\ref{assumptions_on_X} holds. We compare
the two estimating functions given by 
\begin{align*}
g(t,y,x; \theta) &=  (y-(1+2t)x)^2-(\theta+x^2)^{-1}t  \\
h(t,y,x; \theta) &= (\theta+x^2)^{10} (y-(1+2t)x)^2-(\theta+x^2)^9t \,.
\end{align*} 
For both $g$ and $f$ Assumptions \ref{assumptions_on_g} and
\ref{assumptions_for_maintheo} hold, and $g$ is efficient, while $h$
is not. 
\medskip 

Simulations were carried out in the same manner as for the ergodic model. In the
non-ergodic case, an estimator was again found for every sample path when the
efficient estimating function given by $g$ was used. For the inefficient
estimating function given by $h$, there was no solution to the estimating
equation (in $[0.01,1.99]$) in 14\% of the samples for $n=10^4$ and in
39 \% of the samples for $n=10^3$. Figure \ref{fig:qq2} shows QQ-plots of
$\widehat{Z}_{G,n}=\sqrt{n}\,\widehat{W}_{G,n}^{-1}(\hat{\theta}_{G,n}
-\theta_0)$ and $\widehat{Z}_{H,n} = \sqrt{n}\,\widehat{W}_{H,n}^{-1}
(\hat{\theta}_{H,n}-\theta_0)$ compared with a standard normal
distribution, for $n=10^3$ and $n=10^4$ respectively.
\begin{figure}
\begin{center}
\includegraphics[width=0.45\textwidth]{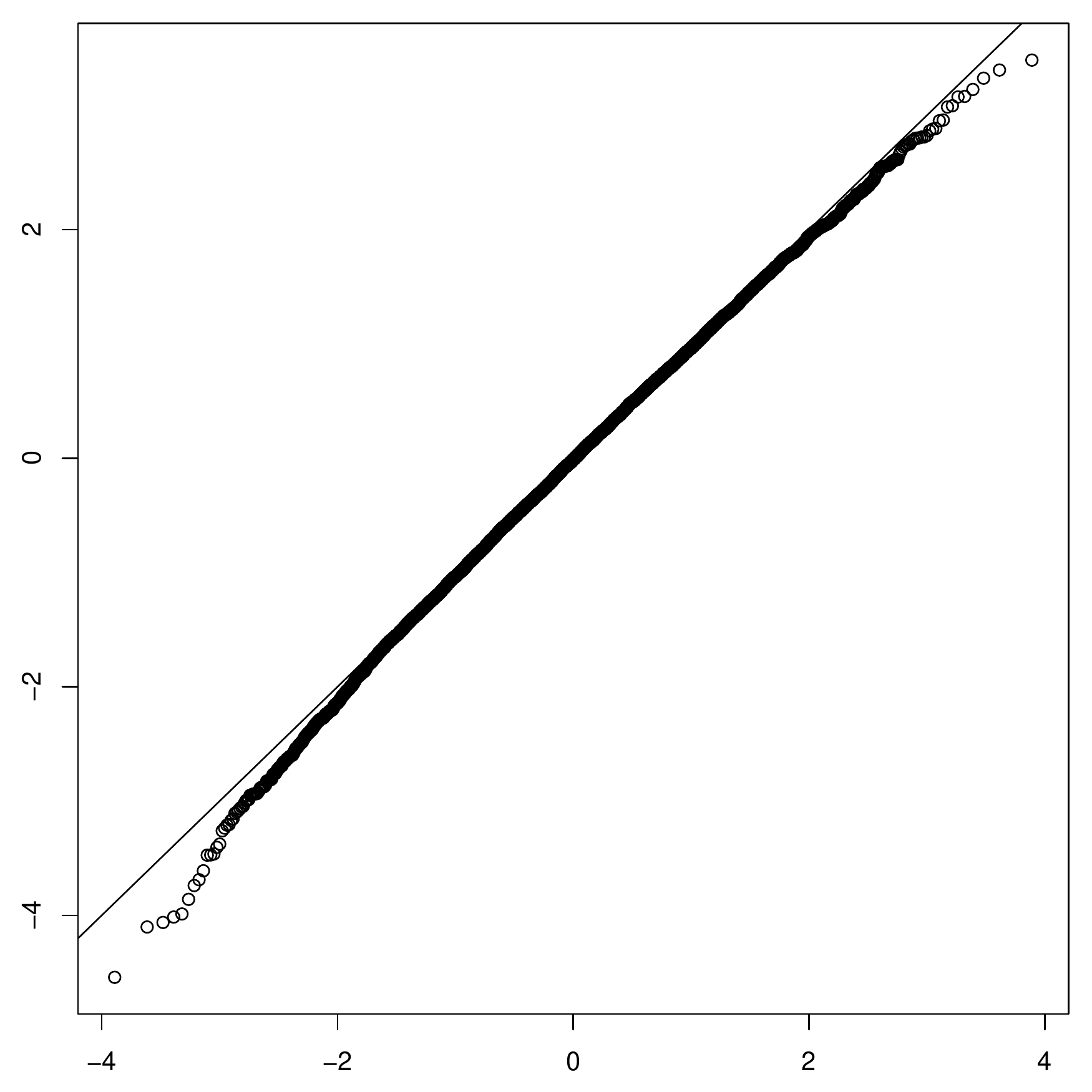}
\includegraphics[width=0.45\textwidth]{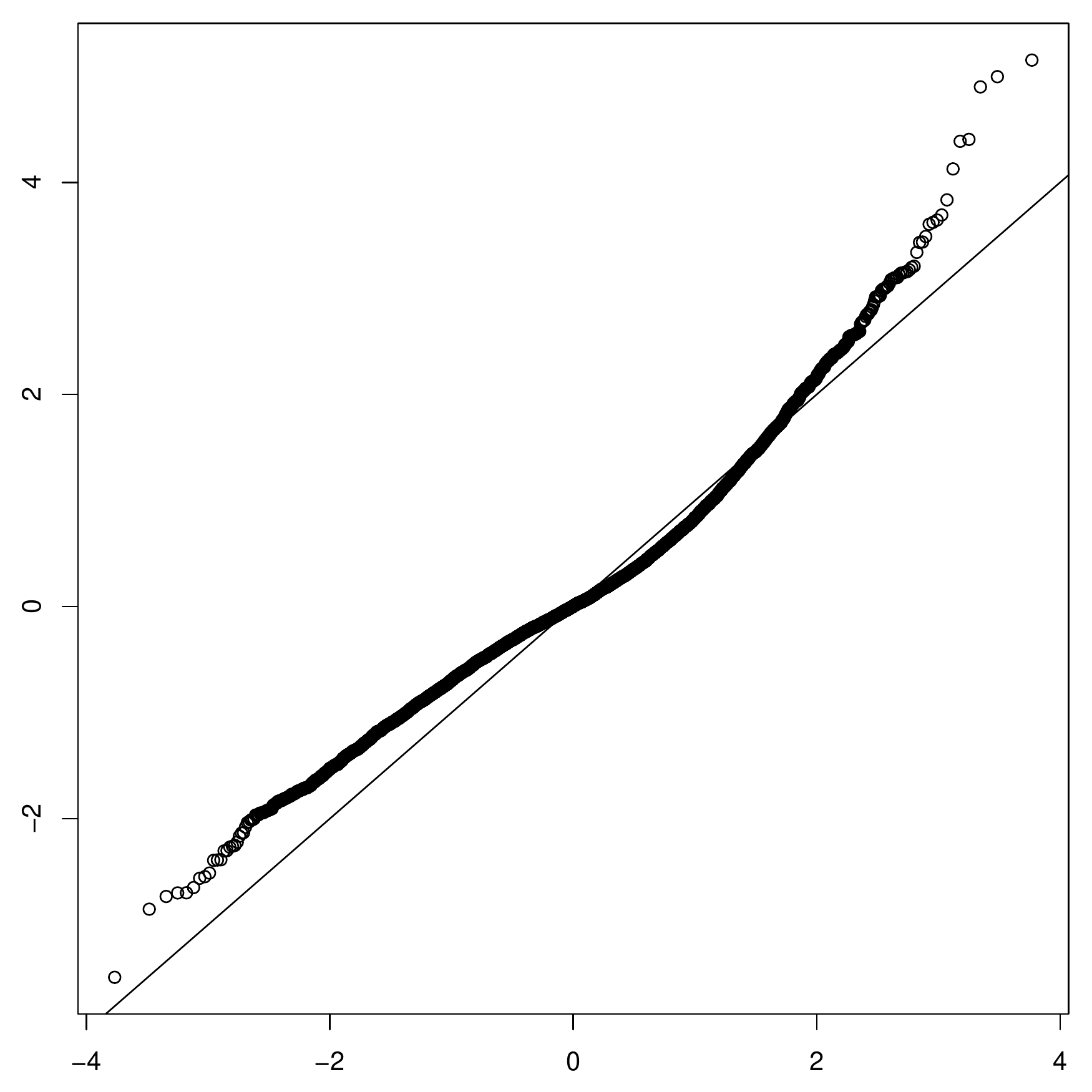}
\includegraphics[width=0.45\textwidth]{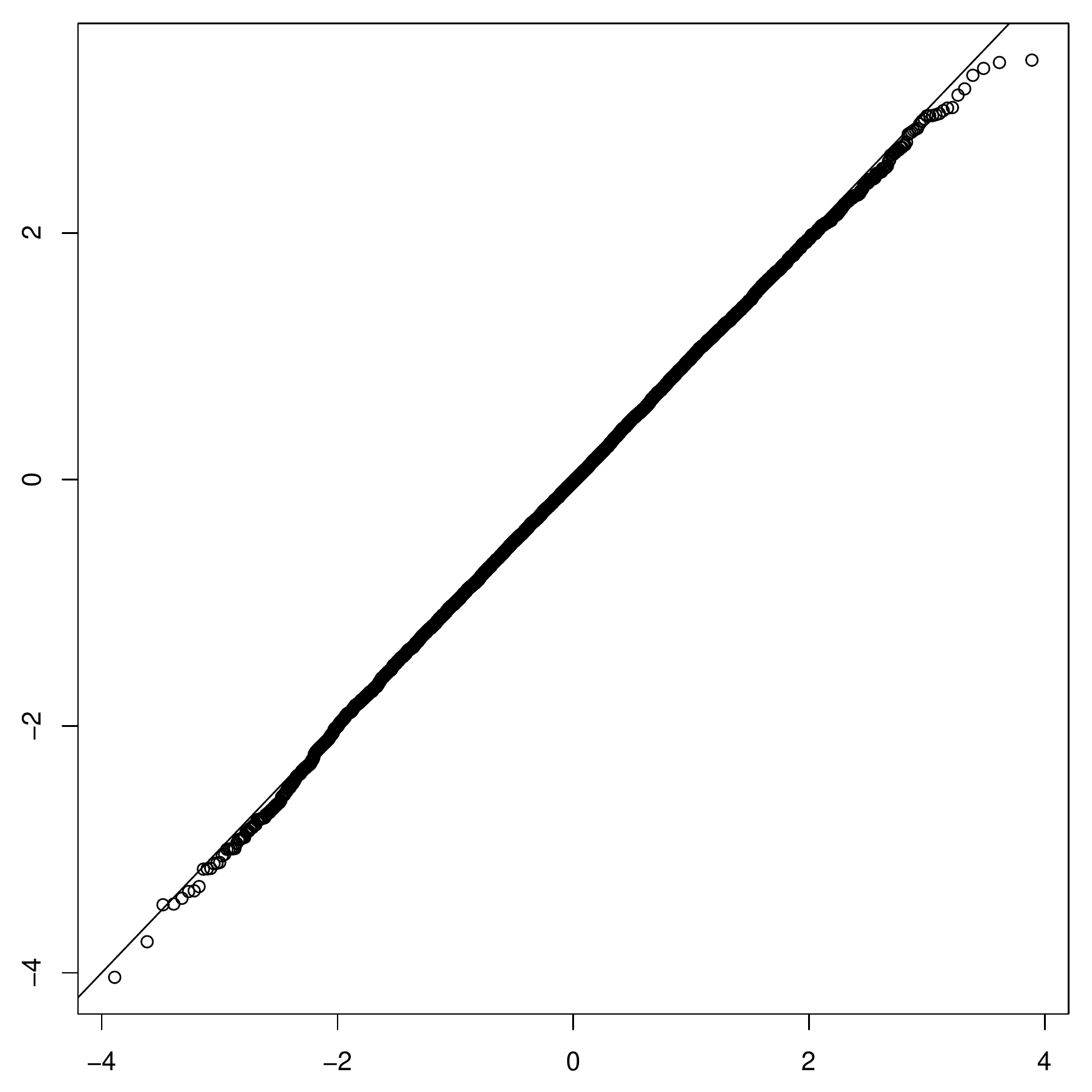}
\includegraphics[width=0.45\textwidth]{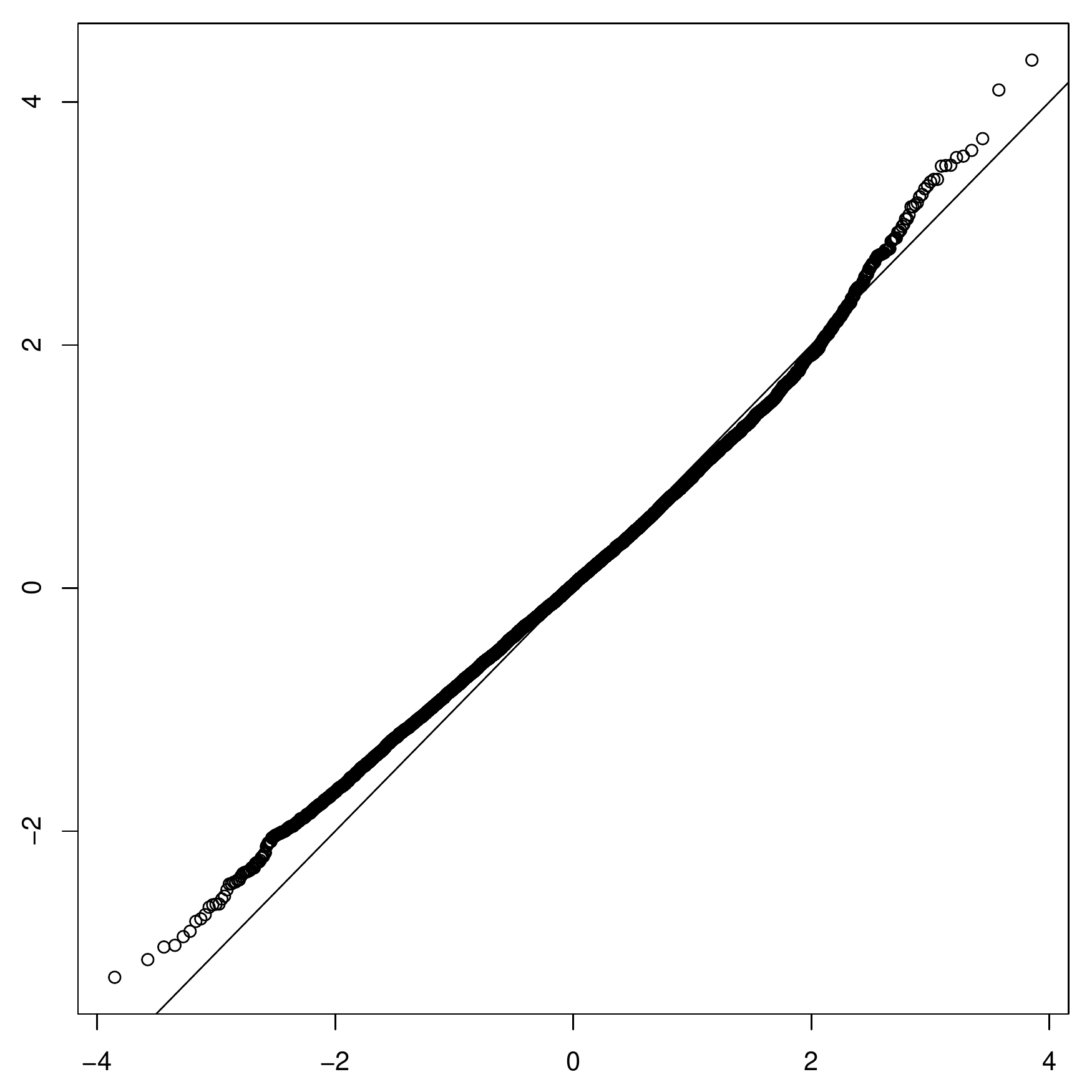}
\end{center}
\caption{QQ-plots comparing $\widehat{Z}_{G,n}$ (left) and
$\widehat{Z}_{H,n}$ (right) to the $\nn(0,1)$ distribution in the case
of the non-ergodic model (\ref{SDEex2}) for
$n=10^3$ (above) and $n=10^4$ (below).} 
\label{fig:qq2}
\end{figure} 
These QQ-plots indicate that in the non-ergodic case there is a
slightly slower convergence to the asymptotic distribution in
Theorem \ref{est_exist}.\ref{item:asymp_estvar} for the efficient
estimating function, and a considerably slower convergence for the
inefficient estimating function, when compared to the ergodic case. 

\section{Proofs}
\label{sec:mainlemmaproofs}
Section \ref{sec:main_lemma} states three main lemmas needed to prove
Theorem \ref{est_exist}, followed by the proof of the theorem.
Section \ref{asymp:lemmaproofs} contains the proofs of the three lemmas.

\subsection{Proof of the Main Theorem}
\label{sec:main_lemma}
In order to prove Theorem \ref{est_exist}, we use the following
lemmas, together with results from \citet{MSandJacod}, and
\citet[Section 1.10]{MSbog}. 
\medskip

\begin{lemma}
Suppose that Assumptions \ref{assumptions_on_X} and
\ref{assumptions_on_g} hold. For $\theta \in \Theta$, let
\begin{align*}
G_n(\theta) &= \sum_{i=1}^n g(\Delta_n, \xtr, \xtl; \theta)\\
G_n^{sq}(\theta) &= \frac{1}{\Delta_n}\sum_{i=1}^n g^2(\Delta_n, \xtr, \xtl; \theta)
\end{align*}
and\
\begin{align*}
A(\theta; \theta_0) &= \tfrac{1}{2} \int_0^1 \left( b^2(X_s;
  \theta_0) - b^2(X_s; \theta)\right) \partial^2_y g(0,X_s,X_s;
  \theta)\, ds\\
B(\theta; \theta_0) &= \tfrac{1}{2} \int_0^1 \left(b^2\left(X_s;\theta_0\right) -
  b^2\left(X_s;
    \theta\right)\right) \partial_y^2 \partial_\theta g(0,X_s,X_s;\theta) \,ds \\
&\hspace{5mm} -
  \tfrac{1}{2}\int_0^1 \partial_\theta b^2(X_s;\theta)\partial_y^2
  g(0,X_s, X_s;\theta)\, ds \\
C(\theta; \theta_0) &= \tfrac{1}{2} \int_0^1\left( b^4(X_s;\theta_0) +
  \tfrac{1}{2}\left( b^2(X_s;\theta_0) -
    b^2(X_s;\theta)\right)^2\right) \left(\partial^2_y
  g(0,X_s,X_s;\theta)\right)^2\, ds\,.
\end{align*}

Then, 
\begin{enumerate}[label=(\roman{*}), ref=(\roman{*})]
\item \label{item:mappings}the mappings $\theta \mapsto A(\theta; \thetan)$,  $\theta
  \mapsto B(\theta; \thetan)$ and  $\theta \mapsto C(\theta; \thetan)$
  are continuous on $\Theta$ ($\PP_\thetan$-almost surely) with $A(\thetan; \thetan)=0$ and
  $\partial_\theta A(\theta; \thetan) = B(\theta; \thetan)$.
\item\label{item:pointwise} for all
$t\in[0,1]$,
\begin{align}
\frac{1}{\sqrt{\Delta_n}}\sum_{i=1}^{[nt]} \left| \EE_{\theta_0}\left( g(\Delta_n, X_{t_i^n},
  X_{t_{i-1}^n}; \theta_0) \mid \xtl\right) \right|&\overset{\pp}{\longrightarrow}
0
\label{uni_prob_g_1}\\
\frac{1}{\Delta_n}\sum_{i=1}^{[nt]} \left( \EE_{\theta_0}\left( g(\Delta_n, X_{t_i^n},
  X_{t_{i-1}^n}; \theta_0) \mid \xtl\right)\right)^2&\overset{\pp}{\longrightarrow} 0
\label{uni_prob_g_2} \\
\frac{1}{\Delta_n^2}\sum_{i=1}^{[nt]} \EE_{\theta_0}\left( g^4(\Delta_n, X_{t_i^n},
  X_{t_{i-1}^n}; \theta_0) \mid \xtl\right)& \overset{\pp}{\longrightarrow} 0
\label{uni_prob_squared_g_2}
\end{align}
and
\begin{align}
\begin{split}\label{uni_prob_squared_g_1}
\frac{1}{\Delta_n}\sum_{i=1}^{[nt]}   \EE_{\theta_0}\left( g^2(\Delta_n, X_{t_i^n},
  X_{t_{i-1}^n}; \theta_0) \mid \xtl\right) \overset{\pp}{\longrightarrow} \tfrac{1}{2} \int_0^{t}b^4(X_s;\theta_0)  \left(\partial^2_y
  g(0,X_s,X_s;\theta_0)\right)^2\, ds\,.
\end{split}
\end{align}
\item\label{item:uniform} for all compact, convex subsets $K \subseteq \Theta$,
\begin{align*}
\sup_{\theta \in K}\left|G_{n}(\theta) - A(\theta; \theta_0)\right|
&\overset{\pp}{\longrightarrow} 0
\nonumber\\
\sup_{\theta \in K} \left|\partial_\theta G_{n}(\theta) -
  B(\theta;\theta_0)\right| &\overset{\pp}{\longrightarrow} 0
\\
\sup_{\theta \in K} \left|G_{n}^{sq}(\theta) -
  C(\theta;\theta_0)\right| &\overset{\pp}{\longrightarrow} 0\,.
\end{align*}
\end{enumerate}
\theoqed
\label{main_lemma_P}
\end{lemma}
\begin{lemma}
Suppose that Assumptions \ref{assumptions_on_X} and
\ref{assumptions_on_g} hold. Then, for all $t\in[0,1]$,
\begin{align}
\frac{1}{\sqrt{\Delta_n}} \sum_{i=1}^{[nt]} \EE_{\theta_0}\left(
  g(\Delta_n, \xtr, \xtl; \theta_0) ( \wtr - \wtl)\mid 
  \ftl \right) \overset{\pp}{\longrightarrow} 0\,.
\label{please_converge}
\end{align}\theoqed
\label{g_delta_w_1}
\end{lemma}

\begin{lemma}
Suppose that Assumptions \ref{assumptions_on_X} and
\ref{assumptions_on_g} hold, and let 
\begin{align*}
Y_{n,t} &= \frac{1}{\sqrt{\Delta_n}} \sum_{i=1}^{[nt]} g(\Delta_n,\xtr, \xtl;
\theta_0)\,.
\end{align*}
Then the sequence of processes
$(\mathbf{Y}_n)_{n\in \NN}$ given by $\mathbf{Y}_n = (Y_{n,t})_{t\in [0,1]}$ 
converges stably in distribution under $\PP_\thetan$ to the process $\mathbf{Y}  =
(Y_t)_{t\in [0,1]}$
given by
\begin{align*}
Y_t &=  \tfrac{1}{\sqrt{2}}\int_0^t b^2(X_s;
\theta_0) \partial^2_y g(0,X_s,X_s; \theta_0)\,dB_s\,.
\end{align*}
Here $\B = (B_s)_{s\geq 0}$ denotes a standard Wiener process, which
is defined on a filtered extension $(\Omega', \ff', (\ff_t')_{t\geq 0}, 
P_\thetan')$ of $(\Omega, \ff, (\ff_t)_{t\geq 0},
P_\thetan)$, and is independent of $(U, \W)$.
\label{conv_dist}
\theoqed\end{lemma}
We denote stable convergence in distribution under $\PP_{\theta_0}$
as $n\to\infty$ by $\overset{\dd_{st}}{\longrightarrow}$.
\medskip

\begin{proof}[\textbf{Proof of Theorem \ref{est_exist}}]
Let a compact, convex subset $K\subseteq \Theta$ with $\thetan \in
\text{int}\, K$ be given.
The functions $G_n(\theta)$, $A(\theta,\theta_0)$, $B(\theta,\theta_0)$, and
$C(\theta,\theta_0)$ were defined in Lemma \ref{main_lemma_P}.
\medskip

By Lemma \ref{main_lemma_P}.\ref{item:mappings} and \ref{item:uniform},
\begin{align}
G_n(\theta_0) \overset{\pp}{\longrightarrow}  0\quad\text{ and }\quad
  \sup_{\theta \in K} \left|\partial_\theta G_{n}(\theta) - 
  B(\theta,\theta_0)\right| \overset{\pp}{\longrightarrow} 0
\label{c1}
\end{align}
with $B(\theta_0; \theta_0) \neq 0$ by Assumption
\ref{assumptions_for_maintheo}.\ref{item:bnot0}, so $G_n(\theta)$ 
satisfies the conditions of Theorem 1.58 in \citet{MSbog}.
\medskip

Now, we show (1.161) of Theorem 1.59 in \citet{MSbog}. Let $\varepsilon>0$ be given, and let $\bar{B}_\varepsilon(\thetan)$ and $B_\varepsilon(\thetan)$,
respectively, denote closed and open balls in $\RR$ with
radius $\varepsilon>0$, centered at $\thetan$. The compact set $K
\backslash B_\varepsilon(\theta_0)$ does
not contain $\theta_0$, and so, by Assumption
\ref{assumptions_for_maintheo}.\ref{item:anot0},
$A(\theta,\theta_0)\neq 
0$ for all $\theta \in K 
\backslash B_\varepsilon(\theta_0)$ with probability one under $\PP_\thetan$. 
\medskip

Because
\begin{align*}
\inf_{\theta \in K
\backslash \bar{B}_\varepsilon(\theta_0)} |A(\theta,\theta_0)|& \geq
\inf_{\theta \in K
\backslash B_\varepsilon(\theta_0)} |A(\theta,\theta_0)| > 0
\end{align*}
$\PP_\thetan$-almost surely, by the continuity of $\theta \mapsto A(\theta, \theta_0)$,
it follows that
\begin{align*}
\PP_{\theta_0}\left( \inf_{\theta \in K\backslash \bar{B}_\varepsilon(\theta_0)}
|A(\theta,\theta_0)|>0\right) = 1\,.
\end{align*}
Consequently, by Theorem 1.59 in \citet{MSbog}, for any
$G_n$-estimator $\tilde{\theta}_n$,
\begin{align}
\PP_{\theta_0}\left( \tilde{\theta}_n \in K\backslash
  \bar{B}_\varepsilon(\theta_0)\right) \to 0\quad\text{ as }\quad n\to \infty\,.
\label{propertyi}\end{align}
for any $\varepsilon > 0$.
\medskip

By Theorem 1.58 in \citet{MSbog}, there exists a consistent $G_n$-estimator
$\hat{\theta}_n$, which is eventually
unique, in the sense that if $\bar{\theta}_n$ is another consistent
$G_n$-estimator, then
\begin{align}
\PP_{\theta_0}\left( \hat{\theta}_n
  \neq \bar{\theta}_n\right) \to 0\quad\text{ as }\quad n\to \infty\,.
\label{unique_est}
\end{align}
Suppose that $\tilde{\theta}_n$ is any
$G_n$-estimator which satisfies that
\begin{align}
\PP_{\theta_0}\left( \tilde{\theta}_n \in
  K\right) \to 1\quad\text{ as }\quad n\to \infty\,.
\label{inK}
\end{align}
Combining (\ref{propertyi}) and (\ref{inK}), it follows that
\begin{align}
\PP_\thetan\left( \tilde{\theta}_n \in
  \bar{B}_\varepsilon(\theta_0)\right) \to 1\quad\text{ as }\quad n\to
  \infty\,,
\label{notinKc}
\end{align}

so $\tilde{\theta}_n$ is consistent. Using (\ref{unique_est}), Theorem
\ref{est_exist}.\ref{item:existence_asymp} follows.
\medskip

To prove Theorem
\ref{est_exist}.\ref{item:main_asymp}, recall that $\Delta_n = 1/n$,
and observe that by Lemma \ref{conv_dist},
\begin{align}
\sqrt{n} G_n(\theta_0) \overset{\dd_{st}}{\longrightarrow} S(\theta_0)
\label{stable_limit}
\end{align}
where
\begin{align*}
S(\theta_0) &= \int_0^1\tfrac{1}{\sqrt{2}} b^2(X_s;
\theta_0) \partial^2_y g(0,X_s,X_s; \theta_0)\,dB_s\,,
\end{align*}
and $\B = (B_s)_{s\in [0,1]}$ is a standard Wiener process,
independent of $(U, \W)$. As $\X$ is then also independent
of $\B$, $S(\theta_0)$ is equal in distribution to
$C(\theta_0; \thetan)^{1/2} Z$, where 
$Z$ is standard normal distributed and independent of
$(X_t)_{t\in [0,1]}$. Note that by Assumption
\ref{assumptions_for_maintheo}.\ref{item:cnot0}, the distribution of
$C(\theta_0; \thetan)^{1/2}Z$ is non-degenerate.
\medskip

Let $\hat{\theta}_n$ be a consistent $G_n$-estimator. By (\ref{c1}),
(\ref{stable_limit}) and properties of stable convergence
(e.g. $(2.3)$ in \citet{jacod1997continuous}),
\begin{align*}
\begin{pmatrix} \sqrt{n} G_n(\theta_0) \\  \partial_\theta G_n(\thetan) \end{pmatrix} \overset{\dd_{st}}{\longrightarrow} \begin{pmatrix}
  S(\theta_0) \\ B(\theta_0; \theta_0) \end{pmatrix}\,.
\end{align*}
Stable convergence in distribution implies weak
convergence, so an application of Theorem 1.60 in
\citet{MSbog} yields
\begin{align}
\sqrt{n}(\hat{\theta}_n - \theta_0) &\overset{\dd}{\longrightarrow}
-B(\theta_0,\theta_0)^{-1}S(\theta_0)\,.
\label{stable_tmp}
\end{align}
The limit is equal in distribution to
$W(\theta_0)Z$, where $W(\theta_0) =
-B(\theta_0,\theta_0)^{-1}C(\theta_0; \thetan)^{1/2}$ and $Z$ is
standard normal distributed and independent of 
$W(\theta_0)$. This completes the proof of Theorem \ref{est_exist}.\ref{item:main_asymp}.
\medskip

Finally, Lemma 2.14 in \citet{MSandJacod} is used to write
\begin{align*}
\sqrt{n}(\hat{\theta}_n - \theta_0) &= - B(\theta_0;
\theta_0)^{-1}\sqrt{n} G_n(\theta_0) + \sqrt{n}|\hat{\theta}_n-\thetan|\varepsilon_n(\theta_0)\,,
\end{align*}
where the last term goes to zero in probability under
$\PP_\thetan$. By the stable continuous mapping theorem, 
(\ref{stable_tmp}) holds with stable
convergence in distribution as well. Lemma
\ref{main_lemma_P}.\ref{item:uniform} may be used to conclude that
$\widehat{W}_n \overset{\pp}{\longrightarrow} 
W(\theta_0)$, so Theorem \ref{est_exist}.\ref{item:asymp_estvar}
follows from the stable version of
(\ref{stable_tmp}) by application of standard results for stable convergence.
\end{proof}

\subsection{Proofs of Main Lemmas}
\label{asymp:lemmaproofs}
This section contains the proofs of Lemmas \ref{main_lemma_P},
\ref{g_delta_w_1} and \ref{conv_dist} in Section
\ref{sec:main_lemma}. A number of technical
results are utilised in the proofs, these results are summarised in
Appendix \ref{sec:asymp:auxres}, some of them with a proof.
\begin{proof}[\textbf{Proof of Lemma \ref{main_lemma_P}}]
First, note that for any $f(x; \theta) \in \cc^{\text{pol}}_{0,0}(\xx
\times \Theta)$ and any compact, convex subset
$K\subseteq \Theta$, there exist
constants $C_K > 0$ such that
\begin{align*}
|f(X_s; \theta)| &\leq C_K(1+|X_s|^{C_K})
\end{align*}
for all $s\in [0,1]$ and $\theta \in \text{int}\,K$. With
probability one under $\PP_\thetan$, for fixed $\omega$, 
$C_K(1+|X_s(\omega)|^{C_K})$ is a continuous function and
therefore Lebesgue-integrable over $[0,1]$. Using this method of
constructing integrable upper bounds, Lemma
\ref{main_lemma_P}.\ref{item:mappings} follows by the usual results
for continuity and differentiability of functions given by integrals.
\medskip 

In the rest of this proof, Lemma \ref{lemma:riemann} and
(\ref{rest_limit}) are repeatedly used without reference.
\medskip 

First, inserting $\theta =\theta_0$ into (\ref{Eg}), it is seen that
\begin{align*}
\frac{1}{\sqrt{\Delta_n}}\sum_{i=1}^{[nt]} \left| \EE_{\theta_0}\left( g(\Delta_n, X_{t_i^n}, X_{t_{i-1}^n}; \theta_0) \mid
  X_{t_{i-1}^n}\right) \right|&= \Delta_n^{3/2}\sum_{i=1}^{[nt]}  R(\Delta_n,
X_{t_{i-1}^n}; \theta_0)  \overset{\pp}{\longrightarrow} 0\\
\frac{1}{\Delta_n}\sum_{i=1}^{[nt]} \left( \EE_{\theta_0}\left(
  g(\Delta_n, X_{t_i^n}, X_{t_{i-1}^n}; \theta_0) \mid
  X_{t_{i-1}^n}\right)\right)^2 &= \Delta_n^{3}\sum_{i=1}^{[nt]} R(\Delta_n,
X_{t_{i-1}^n}; \theta_0) \overset{\pp}{\longrightarrow} 0\,,
\end{align*}
proving (\ref{uni_prob_g_1}) and
(\ref{uni_prob_g_2}). Furthermore, using (\ref{Eg}) and (\ref{Eg2_1}),
\begin{align*}
\sum_{i=1}^{n} \EE_{\theta_0}\left( g(\Delta_n, X_{t_i^n}, X_{t_{i-1}^n}; \theta) \mid
  X_{t_{i-1}^n}\right) 
& \overset{\pp}{\longrightarrow} A(\theta; \thetan)\\
\sum_{i=1}^{n} \EE_{\theta_0}\left( g^2(\Delta_n, X_{t_i^n}, X_{t_{i-1}^n}; \theta) \mid
  X_{t_{i-1}^n}\right)
&
\overset{\pp}{\longrightarrow} 0\,,
\end{align*}
so it follows from Lemma \ref{lemma_gc_jacod} that point-wise for $\theta \in \Theta$,
\begin{align}
G_{n}(\theta) - A(\theta;\theta_0)& \overset{\pp}{\longrightarrow} 0
\label{convA}\,.
\end{align}
Using (\ref{Eg2_1}) and
(\ref{Eg4}),
\begin{align*}
&\hspace{-5mm} \frac{1}{\Delta_n}\sum_{i=1}^{[nt]} \EE_{\theta_0}\left( g^2(\Delta_n,X_{t_i^n},X_{t_{i-1}^n}; \theta)\mid
  X_{t_{i-1}^n}\right)  \\
&\overset{\pp}{\longrightarrow} \tfrac{1}{2} \int_0^{t} \left(
  b^4(X_s;\theta_0)  + \tfrac{1}{2}\left( b^2(X_s; \theta_0)
    - b^2(X_s; \theta)\right)^2\right) \left(\partial^2_y
  g(0,X_s,X_s;\theta)\right)^2\, ds
\end{align*}
and
\begin{align*}
\frac{1}{\Delta_n^2}\sum_{i=1}^{[nt]} \EE_{\theta_0}\left( g^4(\Delta_n,X_{t_i^n},X_{t_{i-1}^n}; \theta)\mid
  X_{t_{i-1}^n}\right) &\overset{\pp}{\longrightarrow } 0\,,
\end{align*}
completing the proof of Lemma
\ref{main_lemma_P}.\ref{item:pointwise} when $\theta = \theta_0$
is inserted, and yielding
\begin{align}
G_n^{sq}(\theta) - C(\theta; \theta_0) \overset{\pp}{\longrightarrow} 0
\label{convC}
\end{align}
point-wise for $\theta \in \Theta$ by Lemma \ref{lemma_gc_jacod}, when
$t=1$ is inserted. Also, using (\ref{Edg}) and (\ref{Edg2}), 
\begin{align*}
\sum_{i=1}^{n} \EE_{\theta_0}\left( \partial_\theta
  g(\Delta_n,X_{t_i^n},X_{t_{i-1}^n};\theta) \mid X_{t_{i-1}^n}\right)
&\overset{\pp}{\longrightarrow} B(\theta; \thetan) \\
\sum_{i=1}^{n} \EE_{\theta_0}\left( \left(\partial_\theta
    g(\Delta_n,X_{t_i^n},X_{t_{i-1}^n};\theta)\right)^2\mid
  X_{t_{i-1}^n}\right) &\overset{\pp}{\longrightarrow} 0\,.
\end{align*}
Thus, by Lemma \ref{lemma_gc_jacod}, also
\begin{align}
\partial_\theta G_n(\theta) - B(\theta;
\theta_0) \overset{\pp}{\longrightarrow} 0\,,
\label{convB}
\end{align}
point-wise for $\theta \in \Theta$. Finally, recall that $\partial^j_y
g(0,x,x;
\theta)=0$ for $j=0,1$. Then, using Lemmas
\ref{lemma_moment_diff} and \ref{lemma_limit_moment_diff}, it follows
that for each $m\in \NN$ and compact, convex subset $K\subseteq
\Theta$, there exist constants $C_{m,K} > 0$
such that for all $\theta, \theta' \in K$ and $n\in \NN$,
\begin{align}
\begin{split}
\EE_{\theta_0}|(G_{n}(\theta) - A(\theta;\theta_0)) -
(G_{n}(\theta') - A(\theta';\theta_0))|^{2m} &\leq C_{m,K} \,\,
|\theta-\theta'|^{2m} \label{ABineq}\\
\EE_{\theta_0}|(\partial_\theta G_n(\theta) - B(\theta;
\theta_0)) -
(\partial_\theta G_n(\theta') - B(\theta';
\theta_0))|^{2m} &\leq C_{m,K} \,\, |\theta-\theta'|^{2m} \\
\EE_{\theta_0}|(G_n^{sq}(\theta) - C(\theta;
\theta_0)) -
( G_n^{sq}(\theta') - C(\theta';
\theta_0))|^{2m} &\leq C_{m,K} \,\, |\theta-\theta'|^{2m} 
\,.
\end{split}
\end{align}
By Lemma \ref{main_lemma_P}.\ref{item:mappings}, the functions $\theta \mapsto G_n(\theta)-A(\theta; \theta_0)$,
$\theta \mapsto \partial_\theta G_n(\theta)-B(\theta;
\theta_0)$ and $\theta \mapsto G_n^{sq}(\theta)-C(\theta, \theta_0)$ are continuous on $\Theta$. Thus, using Lemma \ref{uniform_convergence}
together with (\ref{convA}), (\ref{convC}), (\ref{convB}) and
(\ref{ABineq}) completes the proof of Lemma
\ref{main_lemma_P}.\ref{item:uniform}.
\end{proof}
\begin{proof}[\textbf{Proof of Lemma \ref{g_delta_w_1}}]
The overall strategy in this proof is to expand the expression on the
left-hand side of (\ref{please_converge}) in such a manner that all terms
either converge to $0$ by Lemma \ref{lemma:riemann}, or are equal to
$0$ by the martingale properties of stochastic integral terms obtained
by use of It\^{o}'s formula.
\medskip

By Assumption \ref{assumptions_on_g} and  Lemma
\ref{lemma_properties}, the formulae
\begin{equation}
\begin{aligned}
g(0,y,x; \theta) &= \tfrac{1}{2}(y-x)^2 \partial_y^2 g(0,x,x; \theta) + (y-x)^3R(y,x; \theta) \\
g^{(1)}(y,x; \theta) &= g^{(1)}(x,x; \theta) + (y-x)R(y,x; \theta)
\end{aligned}
\label{eqn:g0_taylor}
\end{equation}
may be obtained. Using (\ref{eqn:g_taylor}) and (\ref{eqn:g0_taylor}),
\begin{equation}
\begin{split}
&\hspace{-5mm} \EE_{\theta_0}\left(
  g(\Delta_n, \xtr, \xtl; \theta_0) ( \wtr - \wtl)\mid 
  \ftl\right) \\ 
&= \EE_{\theta_0}\left( \tfrac{1}{2}(\xtr-\xtl)^2\partial_y^2 g(0,\xtl, \xtl; \theta_0) ( \wtr - \wtl)\mid 
  \ftl\right)  \\
&\hspace{5mm} + \EE_{\theta_0}\left(
  (\xtr-\xtl)^3R(\xtr,\xtl; \theta_0) ( \wtr - \wtl)\mid 
  \ftl\right)  \\
&\hspace{5mm} + \Delta_n \EE_{\theta_0}\left(
 g^{(1)}(\xtl,\xtl; \theta_0)  ( \wtr - \wtl)\mid 
  \ftl\right) \\
&\hspace{5mm} +\Delta_n \EE_{\theta_0}\left(
    (\xtr-\xtl)R(\xtr,\xtl; \theta_0)  ( \wtr - \wtl)\mid 
  \ftl\right) \\
&\hspace{5mm} +\Delta^2 \EE_{\theta_0}\left(
    R(\Delta_n,\xtr,\xtl; \theta_0)  ( \wtr - \wtl)\mid 
  \ftl\right)\,.
\end{split}
\label{W_diff_assumption}
\end{equation}
Note that
\begin{align*}
\Delta_n g^{(1)}(\xtl,\xtl; \theta_0) \EE_{\theta_0}\left(
  \wtr - \wtl\mid 
  \ftl\right) &= 0\,,
\end{align*}
and that by repeated use of the Cauchy-Schwarz inequality,
Lemma \ref{diff_growth} and Corollary \ref{f_growth},
\begin{align*}
\left| \EE_{\theta_0}\left(
  (\xtr-\xtl)^3R(\xtr,\xtl; \theta_0) ( \wtr - \wtl)\mid 
  \ftl\right) \right|
&\leq \Delta_n^{2} C ( 1 + |\xtl|^C) \\
\Delta_n \left| \EE_{\theta_0}\left(
    (\xtr-\xtl)R(\xtr,\xtl; \theta_0)  ( \wtr - \wtl)\mid 
  \ftl\right) \right| 
&\leq \Delta_n^{2} C ( 1 + |\xtl|^C)\\
\Delta_n^2  \left|\EE_{\theta_0}\left(
    R(\Delta_n,\xtr,\xtl; \theta_0)  ( \wtr - \wtl)\mid 
  \ftl\right)\right| 
&\leq \Delta_n^{5/2} C( 1 + |\xtl|^C)
\end{align*}
for suitable constants $C>0$,
with
\begin{align*}
\frac{1}{\sqrt{\Delta_n}} \sum_{i=1}^{[nt]} \Delta_n^{m/2} C( 1 +
  |\xtl|^C) &\overset{\pp}{\longrightarrow} 0
\end{align*}
for $m=4,5$ by Lemma \ref{lemma:riemann}. Now, by
(\ref{W_diff_assumption}), it only remains to show that 
\begin{align}
\frac{1}{\sqrt{\Delta_n}} \sum_{i=1}^{[nt]}  \partial_y^2 g(0,\xtl,
\xtl; \theta_0)  \EE_{\theta_0}\left( (\xtr-\xtl)^2( \wtr - \wtl)\mid 
  \ftl\right) \overset{\pp}{\longrightarrow} 0\,.
\label{remaining_prob}
\end{align}
Applying  It\^{o}'s formula with the function
\begin{align*}
f(y,w) &= (y-x_{t_{i-1}^n})^2(w- w_{t_{i-1}^n})
\end{align*}
to the process $(X_t,W_t)_{t\geq t_{i-1}^n}$, conditioned on $(\xtl,
\wtl) = (x_{t_{i-1}^n}, w_{t_{i-1}^n})$, it follows that
\begin{equation}
\begin{split}
& \hspace{-5mm} (\xtr-\xtl)^2(\wtr- \wtl) \\
&=2 \int_{t_{i-1}^n}^{t_i^n} (X_s-\xtl)(W_s-
\wtl)a(X_s)\, ds+ \int_{t_{i-1}^n}^{t_i^n}(W_s- \wtl)
b^2(X_s; \theta_0)\,ds \\
&\hspace{5mm} + 2\int_{t_{i-1}^n}^{t_i^n}(X_s-\xtl) b(X_s;
\theta_0)\, ds + 2\int_{t_{i-1}^n}^{t_i^n}(X_s-\xtl)(W_s- \wtl)b(X_s;
\theta_0)\, dW_s \\
&\hspace{5mm} + \int_{t_{i-1}^n}^{t_i^n} (X_s-\xtl)^2\, dW_s\,.
\end{split}
\label{Ito1}
\end{equation}
By the martingale property of the It\^{o} integrals in (\ref{Ito1}), 
\begin{align}
\begin{split}\label{almost_converged}
& \hspace{-5mm} \EE_{\theta_0}\left( (\xtr-\xtl)^2( \wtr - \wtl)\mid 
  \ftl\right) \\
&=  2 \int_{t_{i-1}^n}^{t_i^n} \EE_{\theta_0}\left((X_s-\xtl)(W_s-
\wtl)a(X_s) \mid 
  \ftl\right) \, ds\\
&\hspace{5mm} +\int_{t_{i-1}^n}^{t_i^n} \EE_{\theta_0}\left(  (W_s- \wtl)
b^2(X_s; \theta_0) \mid 
  \ftl\right)\,ds  \\
&\hspace{5mm} + 2\int_{t_{i-1}^n}^{t_i^n} \EE_{\theta_0}\left(  (X_s-\xtl) b(X_s;
\theta_0) \mid 
  \xtl  \right)\, ds\,.
\end{split}
\end{align}
Using the Cauchy-Schwarz inequality, Lemma \ref{diff_growth} and Corollary
\ref{f_growth} again, 
\begin{align*}
\left|\int_{t_{i-1}^n}^{t_i^n} \EE_{\theta_0}\left((X_s-\xtl)(W_s-
\wtl)a(X_s) \mid 
  \ftl\right) \, ds\right|
&\leq C\Delta_n^2 ( 1+|\xtl|^C)\,,
\end{align*}
and by Lemma \ref{stoch_taylor}
\begin{align*}
\EE_{\theta_0}\left(  (X_s-\xtl) b(X_s;
\theta_0) \mid 
  \xtl  \right) &= (s-t_{i-1}^n)R(s-t_{i-1}^n, \xtl; \theta_0)\,,
\end{align*}
so also
\begin{align*}
\left|\int_{t_{i-1}^n}^{t_i^n} \EE_{\theta_0}\left(  (X_s-\xtl) b(X_s;
\theta_0) \mid 
  \xtl  \right)\, ds \right|
&\leq C\Delta_n^2( 1+|\xtl|^C)\,.
\end{align*}
Now
\begin{align*}
&\hspace{-5mm}\left|\frac{1}{\sqrt{\Delta_n}} \sum_{i=1}^{[nt]}  \partial_y^2 g(0,\xtl,
\xtl; \theta_0) \int_{t_{i-1}^n}^{t_i^n} \EE_{\theta_0}\left((X_s-\xtl)(W_s-
\wtl)a(X_s) \mid 
  \ftl\right) \, ds\right| \\
&+\left|\frac{1}{\sqrt{\Delta_n}} \sum_{i=1}^{[nt]}  \partial_y^2 g(0,\xtl,
\xtl; \theta_0) \int_{t_{i-1}^n}^{t_i^n} \EE_{\theta_0}\left(  (X_s-\xtl) b(X_s;
\theta_0) \mid 
  \xtl  \right)\, ds \right|\\
&\leq \Delta_n^{3/2}C\sum_{i=1}^{[nt]}  \left|\partial_y^2 g(0,\xtl,
\xtl; \theta_0) \right| ( 1+|\xtl|^C) 
\overset{\pp}{\longrightarrow} 0
\end{align*}
by Lemma \ref{lemma:riemann}, so by (\ref{remaining_prob}) and
(\ref{almost_converged}), it remains to show that
\begin{align*}
\frac{1}{\sqrt{\Delta_n}} \sum_{i=1}^{[nt]}  \partial_y^2 g(0,\xtl,
\xtl; \theta_0)  \int_{t_{i-1}^n}^{t_i^n} \EE_{\theta_0}\left(  (W_s- \wtl)
b^2(X_s; \theta_0) \mid 
  \ftl\right)\,ds &\overset{\pp}{\longrightarrow} 0\,.
\end{align*}
Applying It\^{o}'s formula with the function
\begin{align*}
f(y,w) &= (w- w_{t_{i-1}^n}) b^2(y; \theta_0)\,,
\end{align*}
and making use of the martingale properties of the stochastic integral terms, yields
\begin{align*}
& \int_{t_{i-1}^n}^{t_i^n}\EE_{\theta_0}\left( (W_s-\wtl) b^2(X_s; \theta_0)
  \mid \ff_{t_{i-1}^n}\right)\, ds\\
&= \int_{t_{i-1}^n}^{t_i^n}\int_{t_{i-1}^n}^s\EE_{\theta_0}\left( a(X_u) \partial_yb^2(X_u; \theta_0) (W_u-
\wtl)\mid \ff_{t_{i-1}^n}\right)\, du\,ds\\
&\hspace{5mm} + \tfrac{1}{2}\int_{t_{i-1}^n}^{t_i^n}\int_{t_{i-1}^n}^s\EE_{\theta_0}\left(b^2(X_u; \theta_0)  \partial^2_y b^2(X_u;
    \theta_0)  (W_u- \wtl)\mid
\ff_{t_{i-1}^n}\right)\,du\,ds\\
&\hspace{5mm} + \int_{t_{i-1}^n}^{t_i^n}\int_{t_{i-1}^n}^s\EE_{\theta_0}\left(b(X_u;\theta_0)\partial_yb^2(X_u;
\theta_0) \mid \ff_{t_{i-1}^n}\right)\, du\,ds\,.
\end{align*}
Again, by repeated use of the Cauchy-Schwarz inequality and
Corollary \ref{f_growth},
\begin{align*}
 \left|\int_{t_{i-1}^n}^{t_i^n}\EE_{\theta_0}\left(  (\wtr-\wtl) b^2(X_s; \theta_0)
  \mid \ff_{t_{i-1}^n}\right)\, ds \right|
&\leq C( 1 + |\xtl|^C)( \Delta_n^2 + \Delta_n^{5/2})\,.
\end{align*}
Now
\begin{align*}
&\hspace{-5mm} \left| \frac{1}{\sqrt{\Delta_n}} \sum_{i=1}^{[nt]}  \partial_y^2 g(0,\xtl,
\xtl; \theta_0)  \int_{t_{i-1}^n}^{t_i^n} \EE_{\theta_0}\left(  
 (W_s- \wtl) b^2(X_s; \theta_0)\mid 
  \ftl\right)\,ds\right| \\
&\leq \left( \Delta_n^{3/2} + \Delta_n^2\right) \sum_{i=1}^{[nt]}  \left| \partial_y^2 g(0,\xtl,
\xtl; \theta_0) \right| C( 1 + |\xtl|^C)
\overset{\pp}{\longrightarrow} 0\,,
\end{align*}
thus completing the proof.
\end{proof}
\begin{proof}[\textbf{Proof of Lemma \ref{conv_dist}}]
The aim of this proof is to establish that the conditions of Theorem
IX.7.28 in \citet{JS_2003} hold, by which the desired result
follows directly.
\medskip

For all $t\in[0,1]$,
\begin{align*}
\sup_{s\leq t} \left| \frac{1}{\sqrt{\Delta_n}}\sum_{i=1}^{[ns]} \EE_{\theta_0}\left( g(\Delta_n, X_{t_i^n},
  X_{t_{i-1}^n}; \theta_0) \mid \xtl\right) \right| 
&\leq \frac{1}{\sqrt{\Delta_n}}\sum_{i=1}^{[nt]}\left|  \EE_{\theta_0}\left( g(\Delta_n, X_{t_i^n},
  X_{t_{i-1}^n}; \theta_0) \mid \xtl\right) \right| 
\end{align*}

and since the right-hand side converges to $0$ in probability under
$\PP_{\theta_0}$ by (\ref{uni_prob_g_1}) of Lemma \ref{main_lemma_P},
so does the left-hand side, i.e.\ condition 7.27 of Theorem IX.7.28
holds. From (\ref{uni_prob_g_2}) and
(\ref{uni_prob_squared_g_1}) of Lemma \ref{main_lemma_P}, it follows
that for all $t\in[0,1]$,
\begin{align*}
&\hspace{-5mm} \frac{1}{\Delta_n} \sum_{i=1}^{[nt]} \left( \EE_{\theta_0}\left(
    g^2(\Delta_n, \xtr, \xtl; \theta_0)\mid \xtl \right) - \EE_{\theta_0}\left(
    g(\Delta_n, \xtr, \xtl; \theta_0) \mid \xtl\right)^2\right) \nonumber\\
&\overset{\pp}{\longrightarrow} \tfrac{1}{2}\int_0^t  b^4(X_s;
\theta_0)\left(\partial_y^2 g(0,X_s,X_s; \theta_0)\right)^2\, ds\,,
\end{align*}
establishing that condition 7.28 of Theorem IX.7.28 is
satisfied. Lemma \ref{g_delta_w_1} implies condition 7.29, while the
Lyapunov condition (\ref{uni_prob_squared_g_2}) of Lemma
\ref{main_lemma_P} implies the Lindeberg condition 7.30 of Theorem
IX.7.28 in \citet{JS_2003}, from which  the desired result now follows. 
\medskip

Theorem IX.7.28 contains an
additional condition 7.31. This condition has
the same form as (\ref{please_converge}), but with $\wtr -
\wtl$ replaced by $N_{t_i^n}-N_{t_{i-1}^n}$, where $\N = (N_t)_{t\geq
  0}$ is any bounded martingale on $(\Omega, \ff, (\ff_t)_{t\geq 0},
\PP_\thetan)$, which is orthogonal to $\W$. However, since
$(\ff_t)_{t\geq 0}$ is generated by U and $\W$, it follows from the
martingale representation theorem \citep[Theorem III.4.33]{JS_2003}
that every martingale on $(\Omega, \ff, (\ff_t)_{t\geq 0},
\PP_\thetan)$ may be written as the sum of a constant term and a stochastic
integral with respect to $\W$, and therefore cannot be orthogonal to
$\W$.
\end{proof}
\appendix

\section{Auxiliary Results}\label{sec:asymp:auxres}
This section contains a number of technical results used in the
proofs in Section \ref{asymp:lemmaproofs}.
 
\begin{lemma}{\cite[Lemma 9]{gc_jacod}}
For $i,n \in \NN$, let $\ff_{n,i} = \ff_{t_i^n}$ (with $\ff_{n,0}=\ff_0$), and
let $F_{n,i}$ be an $\ff_{n,i}$-measurable, real-valued random variable. If
\begin{align*}
\sum_{i=1}^n \EE_{\thetan}(F_{n,i} \mid \ff_{n,i-1}) 
\overset{\pp}{\longrightarrow} F\quad\text{ and }\quad \sum_{i=1}^n \EE_{\thetan}(F_{n,i}^2 \mid \ff_{n,i-1}) 
\overset{\pp}{\longrightarrow} 0\,,
\end{align*}

for some random variable $F$, then
\begin{align*}
\sum_{i=1}^n F_{n,i} \overset{\pp}{\longrightarrow} F\,.
\end{align*}\theoqed
\label{lemma_gc_jacod}
\end{lemma}
\begin{lemma}
Suppose that Assumptions \ref{assumptions_on_X} and
\ref{assumptions_on_g} hold. Then, for all $\theta \in \Theta$, 
\begin{enumerate}[label=(\roman{*}), ref=(\roman{*})]
\item \label{item:Eg}
\begin{align}
\begin{split}
& \hspace{-5mm} \EE_{\theta_0}\left( g(\Delta_n, X_{t_i^n}, X_{t_{i-1}^n}; \theta) \mid
  X_{t_{i-1}^n}\right) \label{Eg}\\
&= \tfrac{1}{2}\Delta_n\left( b^2(\xtl; \theta_0) - b^2(\xtl;
  \theta)\right) \partial_y^2 g(0,\xtl,\xtl; \theta) + \Delta_n^2 R(\Delta_n,
X_{t_{i-1}^n}; \theta)\,,
\end{split}
\end{align}
\item\label{item:Edg}
\begin{align}
\begin{split}
& \hspace{-5mm}  \EE_{\theta_0}\left( \partial_\theta
  g(\Delta_n,\xtr,\xtl;\theta) \mid X_{t_{i-1}^n}\right) \label{Edg}\\
&= \tfrac{1}{2} \Delta_n\left(b^2(\xtl; \theta_0)-b^2(\xtl;
  \theta)\right)\partial^2_y \partial_\theta g(0,\xtl,\xtl; \theta) \\
&\hspace{5mm} - \tfrac{1}{2} \Delta_n \partial_\theta
  b^2(\xtl; \theta)\partial^2_y g(0,\xtl,\xtl; \theta) + \Delta_n^2  R(\Delta_n,
  X_{t_{i-1}^n}; \theta) \,,
\end{split}
\end{align}
\item \label{item:Eg2_1}
\begin{align}
\begin{split}
& \hspace{-5mm} \EE_{\theta_0}\left( g^2(\Delta_n, \xtr, \xtl; \theta)
  \mid X_{t_{i-1}^n}\right) \label{Eg2_1}\\
&= \tfrac{1}{2}\Delta_n^2 \left( b^4(\xtl; \theta_0) +
  \tfrac{1}{2}\left( b^2(\xtl; \theta_0) - b^2(\xtl;
    \theta)\right)^2\right) \left(\partial^2_y
g(0,\xtl,\xtl; \theta)\right)^2 \\
&\hspace{5mm} + \Delta_n^3R(\Delta_n,
\xtl; \theta)\,,
\end{split}
\end{align}
\item\label{item:Edg2}
\begin{align}
\begin{split}
\hspace{-5mm} \EE_{\theta_0}\left( \left(\partial_\theta
    g(\Delta_n,X_{t_i^n},X_{t_{i-1}^n};\theta)\right)^2\mid X_{t_{i-1}^n}\right) 
&= \Delta_n^2  R(\Delta_n,X_{t_{i-1}^n};
  \theta) \,,\label{Edg2} 
\end{split} 
\end{align}
\item\label{item:Eg4}
\begin{align}
\begin{split}
\hspace{-5mm} \EE_{\theta_0}\left( g^4(\Delta_n, \xtr, \xtl; \theta)
  \mid X_{t_{i-1}^n}\right) &= \Delta_n^4 R(\Delta_n, \xtl; \theta) \label{Eg4}\,.
\end{split}
\end{align}
\end{enumerate}
\theoqed
\label{lemma:expectations}
\end{lemma}
\begin{proof}[\textbf{Proof of Lemma \ref{lemma:expectations}}]
The formulae (\ref{Eg}), (\ref{Edg}) and (\ref{Eg2_1}) are implicitly
given in the proofs of \citet[Lemmas $3.2$ \& $3.4$]{efficient}. To
prove the two remaining formulae, note first that using (\ref{L4var}),
Assumption \ref{assumptions_on_g}.\ref{dyg0_cont} and 
Lemma \ref{lemma_properties}, 
\begin{align*}
\ll^i_{\theta_0}( g^4(0;\theta))(x,x) & = 0\,,\quad i= 1,2,3
\\
\ll^i_{\theta_0}(
  g^3(0,\theta)g^{(1)}(\theta))(x,x) &= 0 \,, \quad i = 1,2\\
\ll_{\theta_0}(
  g^2(0,\theta)g^{(1)}(\theta)^2)(x,x) &= 0\\
\ll_{\theta_0}(
  g^3(0,\theta)g^{(2)}(\theta))(x,x) &= 0 \\
\ll_{\theta_0}( \partial_\theta
    g(0,\theta)^2)(x,x) &= 0\,.
\end{align*}
The verification of these formulae may be
simplified by using e.g.\ the Leibniz formula for the $n$'th derivative
of a product to see that partial derivatives are zero when evaluated
in $y=x$. These results, as well as Lemmas \ref{stoch_taylor} and
\ref{lemma_properties}, and (\ref{eqn:expectation_remainder}) are used
without reference in the following.

\begin{align*}
& \hspace{-5mm} \EE_{\theta_0}\left( \left(\partial_\theta
    g(\Delta_n ,X_{t_i^n},X_{t_{i-1}^n};\theta)\right)^2\mid X_{t_{i-1}^n}\right) \\
&= \EE_{\theta_0}\left(\partial_\theta
    g(0,X_{t_i^n},X_{t_{i-1}^n}; \theta)^2\mid
  X_{t_{i-1}^n}\right) \\
&\hspace{5mm} + 2  \Delta_n \EE_{\theta_0}\left(\partial_\theta g(0,X_{t_i^n},X_{t_{i-1}^n};\theta)\partial_\theta g^{(1)}(X_{t_i^n},X_{t_{i-1}^n};\theta)\mid
X_{t_{i-1}^n}\right) \\
&\hspace{5mm} +
\Delta_n^2\EE_{\theta_0}\left(R(\Delta_n,X_{t_i^n},X_{t_{i-1}^n};
  \theta)\mid X_{t_{i-1}^n}\right) \\
&=  \partial_\theta
    g(0,X_{t_{i-1}^n},X_{t_{i-1}^n}; \theta)^2 + \Delta_n
  \ll_{\theta_0}( \partial_\theta
    g(0, \theta)^2) (X_{t_{i-1}^n},X_{t_{i-1}^n}) +\Delta_n^2  R(\Delta_n,X_{t_{i-1}^n};
  \theta) \\ 
&\hspace{5mm} + 2\Delta_n\left( \partial_\theta g(0,X_{t_{i-1}^n},X_{t_{i-1}^n};\theta)\partial_\theta g^{(1)}(X_{t_{i-1}^n},X_{t_{i-1}^n};\theta) + \Delta_n R(\Delta_n,X_{t_{i-1}^n};
  \theta)\right)  \\
&= \Delta_n^2  R(\Delta_n,X_{t_{i-1}^n};
  \theta)\,,
\end{align*}
proving (\ref{Edg2}). Similarly,
\begin{align*}
& \hspace{-5mm} \EE_{\theta_0}\left( g^4(\Delta_n, \xtr, \xtl; \theta)
  \mid X_{t_{i-1}^n}\right) \\
&= \EE_{\theta_0}\left( g^4(0,\xtr,\xtl; \theta) \mid \xtl\right) \\
&\hspace{5mm}+
4\Delta_n   \EE_{\theta_0}\left( g^3(0,\xtr,\xtl; \theta)g^{(1)}(\xtr,\xtl; \theta)\mid
  \xtl\right) \\
&\hspace{5mm} + 6\Delta_n ^2\EE_{\theta_0}\left(  g^2(0,\xtr,\xtl; \theta)g^{(1)}(\xtr,\xtl;
  \theta)^2 \mid
\xtl\right) \\
&\hspace{5mm} + 2\Delta_n^2\EE_{\theta_0}\left( g^3(0,\xtr,\xtl; \theta) g^{(2)}(\xtr,\xtl; \theta)\mid
\xtl\right) \\
&\hspace{5mm} + 4\Delta_n^3 \EE_{\theta_0}\left( g(0,\xtr,\xtl; \theta)g^{(1)}(\xtr,\xtl;
  \theta)^{3} \mid \xtl\right) \\
&\hspace{5mm} + 6\Delta_n^3 \EE_{\theta_0}\left( g^2(0,\xtr,\xtl; \theta)g^{(1)}(\xtr,\xtl;
  \theta)g^{(2)}(\xtr,\xtl;
  \theta) \mid \xtl \right)\\
&\hspace{5mm} + \tfrac{2}{3}\Delta_n^3 \EE_{\theta_0}\left( g^3(0,\xtr,\xtl; \theta)g^{(3)}(\xtr,\xtl;
  \theta)\mid \xtl\right) \\
&\hspace{5mm} + \Delta_n ^4  \EE_{\theta_0}\left( R(\Delta_n ,\xtr,\xtl; \theta)\mid
  \xtl \right) \\
&= g^4(0,\xtl,\xtl; \theta) + \Delta_n \ll_{\theta_0}( g^4(0;
  \theta) )(\xtl,\xtl) + \tfrac{1}{2}\Delta_n^2 \ll^2_{\theta_0}( g^4(0;
  \theta) )(\xtl,\xtl) \\
& \hspace{5mm} + \tfrac{1}{6}\Delta_n^3 \ll^3_{\theta_0}( g^4(0;
  \theta) )(\xtl,\xtl)  + 4\Delta_n g^3(0,\xtl,\xtl;
  \theta)g^{(1)}(\xtl,\xtl; \theta)  \\
&\hspace{5mm} + 4\Delta_n^2 \ll_{\theta_0}( g^3(0;
  \theta)g^{(1)}(\theta))(\xtl, \xtl) + 2\Delta^3_n \ll^2_{\theta_0}( g^3(0;
  \theta)g^{(1)}(\theta))(\xtl, \xtl) \\
&\hspace{5mm} + 6\Delta_n ^2 g^2(0,\xtl,\xtl; \theta)g^{(1)}(\xtl,\xtl;
  \theta)^2 + 6\Delta_n^3 \ll_{\theta_0}( g^2(0;
    \theta)g^{(1)}(\theta)^2) (\xtl,\xtl)\\
&\hspace{5mm} +2\Delta_n^2 g^3(0,\xtl,\xtl; \theta) g^{(2)}(\xtl,\xtl;
\theta) + 2\Delta_n^3 \ll_{\theta_0}(  g^3(0; \theta) g^{(2)}(
\theta) )(\xtl, \xtl) \\
&\hspace{5mm}+ 4\Delta_n^3  g(0,\xtl,\xtl; \theta)g^{(1)}(\xtl,\xtl;
  \theta)^{3}\\
&\hspace{5mm} + 6\Delta_n^3 g^2(0,\xtl,\xtl; \theta)g^{(1)}(\xtl,\xtl;
  \theta)g^{(2)}(\xtl,\xtl;
  \theta) \\
&\hspace{5mm} + \tfrac{2}{3}\Delta_n^3 g^3(0,\xtl,\xtl; \theta)g^{(3)}(\xtl,\xtl;
  \theta) \\
&\hspace{5mm} + \Delta_n^4 R(\Delta_n, \xtl; \theta) \\
&= \Delta_n^4 R(\Delta_n, \xtl; \theta)\,,
\end{align*}
which proves (\ref{Eg4}).
\end{proof}
\begin{lemma}
Let $x \mapsto f(x)$ be a continuous, real-valued
function, and let $t\in[0,1]$ be given. Then
\begin{align*}
\Delta_n\sum_{i=1}^{[nt]} f( X_{t_{i-1}^n})
\overset{\pp}{\longrightarrow} \int_0^{t} f(X_s)\,
ds \,.
\end{align*}
\theoqed
\label{lemma:riemann}
\end{lemma}
Lemma \ref{lemma:riemann} follows easily by the convergence of
Riemann sums.
\begin{lemma}
Suppose that Assumption \ref{assumptions_on_X} holds, and let $m\geq 2$. Then, there exists a constant $C_{m}>0$, such that for $0\leq t \leq t+\Delta \leq 1$,
\begin{align}
\EE_{\theta_0}\left( \vert
  X_{t+\Delta}-X_t\vert^{m}\mid X_t\right) &\leq C_m\Delta^{m/2}\left(
  1+\vert X_t\vert^m\right)\,.
\label{asymp_ineq}
\end{align}\theoqed
\label{diff_growth}
\end{lemma}
\begin{corollary}
Suppose that Assumption \ref{assumptions_on_X} holds. Let a compact, convex
set $K\subseteq \Theta$ be given, and suppose that $f(y,x; \theta)$
is of polynomial growth in $x$ and $y$, uniformly for $\theta$ in $K$. Then, there exist
constants $C_K>0$ such that for $0\leq t \leq t+\Delta \leq 1$,
\begin{align*}
\EE_{\theta_0}\left( |f(X_{t+\Delta}, X_t, \theta)|\mid X_t\right) &\leq
C_K\left(1+|X_t|^{C_K}\right)
\end{align*}
for all $\theta \in K$.\theoqed
\label{f_growth}
\end{corollary}
Lemma \ref{diff_growth} and Corollary \ref{f_growth}, 
correspond to Lemma 6 of \citet{kessler_ergodic}, adapted to the present
assumptions.
For use in the following, observe that for any $\theta \in \Theta$, there
exist constants $C_\theta>0$ such that
\begin{align*}
\Delta_n \sum_{i=1}^{[nt]} \left| R_\theta(\Delta_n, \xtl) \right|
&\leq C_\theta \Delta_n \sum_{i=1}^{[nt]} \left(1 + | \xtl|^{C_\theta}\right) \,,
\end{align*}
so it follows from Lemma \ref{lemma:riemann} that for any deterministic, real-valued
sequence $(\delta_n)_{n\in {\NN}}$ with $\delta_n \to 0$ as $n\to
\infty$, 
\begin{align}
\delta_n \Delta_n \sum_{i=1}^{[nt]} \left|R_\theta(\Delta_n, \xtl)\right| \overset{\pp}{\longrightarrow} 0\,.
\label{rest_limit}
\end{align}
Note that by Corollary \ref{f_growth}, it holds that under
Assumption \ref{assumptions_on_X},
\begin{align}
\EE_{\theta_0}\left( R\left(\Delta, X_{t+\Delta}, X_t;
      \theta\right) \mid X_{t}\right) &= R(\Delta, X_{t}; \theta)\,.
\label{eqn:expectation_remainder}
\end{align}

\begin{lemma}
Suppose that Assumption \ref{assumptions_on_X} holds, and that the function
$f(t,y,x; \theta)$ satisfies that
\begin{align}
f(t,y,x; \theta) \in \cc^{\text{pol}}_{1,2,1}( [0,1]\times \xx^2
  \times \Theta)\quad\text{ with }\quad f(0,x,x;
\theta) = 0
\label{eqn1new}
\end{align}
for all $x\in \xx$ and $\theta \in \Theta$.
Let $m\in \NN$ be given, and let $Dk(\,\cdot\,; \theta,
\theta') = k(\,\cdot\,; \theta) - k(\,\cdot\,; \theta')$. Then, there exist constants $C_m > 0$ such that
\begin{align}
\begin{split}\label{xref}
&\hspace{-5mm} \EE_{\theta_0}\left( \left| Df(t-s, X_t, X_s; \theta,
    \theta')\right|^{2m}\right) \\
&\leq C_{m} (t-s)^{2m-1} \int_s^t \EE_{\theta_0}\left( \left| Df_1(u-s,
    X_u, X_s; \theta, \theta')\right|^{2m}\right)\, du \\
&\hspace{5mm} + C_{m} (t-s)^{m-1} \int_s^t \EE_{\theta_0}\left( \left|
    Df_2(u-s, X_u, X_s; \theta, \theta')\right|^{2m}\right)\, du 
\end{split}
\end{align}
for  $0\leq s < t
\leq 1$ and $\theta, \theta' \in \Theta$, where $f_1$ and $f_2$ are given by
\begin{align*}
f_{1}(t, y,x; \theta) &= \partial_t f\left(t, y,x; \theta\right)
+ a(y)\partial_y f\left(t, y, x; \theta\right) +
\tfrac{1}{2}  b^2(y; \theta_0)\partial^2_y f\left(t, y, x; \theta\right)\\
f_{2}(t, y,x; \theta) &=  b(y; \theta_0)\partial_y f\left(t, y,x ; \theta\right)\,.
\end{align*}
Furthermore, for each
compact, convex set $K\subseteq \Theta$, there exists a constant $C_{m,K}>0$ such that
\begin{align*}
\EE_{\theta_0}\left( |Df_j(t-s, X_t, X_s; \theta, \theta')|^{2m}\right) &\leq C_{m,K} |\theta-\theta'|^{2m}
\end{align*}
for $j=1,2$, $0\leq s < t
\leq 1$ and all $\theta, \theta' \in K$. \theoqed
\label{triv_ineq}
\end{lemma}
\begin{proof}[\textbf{Proof of Lemma \ref{triv_ineq}}]
A simple application of  It\^{o}'s formula (when conditioning on $X_s =
x_s$) yields that for all $\theta \in \Theta$, 
\begin{align}
f(t-s, X_t, X_s; \theta) &=
\int_{s}^{t} f_{1}\left(u-s, X_u, X_{s};
  \theta\right)\, du + \int_{s}^{t}
f_{2}\left(u-s, X_u, X_s; \theta\right)\, dW_u
\label{itof}
\end{align}
under $\PP_{\theta_0}$.
\medskip

By Jensen's inequality, it holds that for any $k\in \NN$,
\begin{align}
\EE_{\theta_0}\left( \left| \int_s^t Df_j(u-s, X_u, X_s; \theta,
    \theta')^j\, du\right|^{k} \right) 
&\leq (t-s)^{k-1} \int_s^t \EE_{\theta_0}\left( \left|  Df_j(u-s, X_u, X_s; \theta, \theta')\right|^{jk}\right)\, du\label{f1ineq}
\end{align}
for $j=1,2$, and by the martingale properties of the second term in
(\ref{itof}), the Burkholder-Davis-Gundy inequality may be used to
show that
\begin{align}
\EE_{\theta_0}\left( \left| \int_s^t Df_2(u-s, X_u, X_s; \theta,
    \theta')\, dW_u\right|^{2m}\right)
\leq C_{m}\EE_{\theta_0}\left( \left| \int_s^t Df_2(u-s, X_u, X_s; \theta,
    \theta')^2\, du\right|^{m}\right)\,.\label{f2ineq}
\end{align}
Now, (\ref{itof}), (\ref{f1ineq}) and (\ref{f2ineq}) may be combined to
show (\ref{xref}).
The last result of the lemma follows by an application of the mean value
theorem.
\end{proof}
\begin{lemma}
Suppose that Assumption \ref{assumptions_on_X} holds, and let
$K\subseteq \Theta$ be compact and convex. Assume that 
$f(t,y,x; \theta)$ satisfies (\ref{eqn1new}) for all $x \in \xx$ and
$\theta \in \Theta$, and define
\begin{align*}
F_n(\theta) &= \sum_{i=1}^{n} f(\Delta_n, X_{t_i^n}, X_{t_{i-1}^n}; \theta)\,.
\end{align*}
Then, for each $m\in \NN$, there exists a constant $C_{m,K}>0$, such that 
\begin{align*}
\EE_{\theta_0}\left| F_n(\theta) -
  F_n(\theta')\right|^{2m} &\leq C_{m,K}\, |\theta-\theta '|^{2m} 
\end{align*}
for all $\theta, \theta' \in K$ and $n\in\NN$. Define
$\widetilde{F}_n(\theta) = \Delta_n^{-1} F_n(\theta)$, and suppose,
moreover, that the functions 
\begin{align*}
h_{1}(t, y,x; \theta) &= \partial_t f\left(t, y,x; \theta\right)
+ a(y)\partial_y f\left(t, y, x; \theta\right) +
\tfrac{1}{2} b^2(y; \theta_0)\partial^2_y f\left(t, y, x; \theta\right)\\
h_{2}(t, y,x; \theta) &= b(y; \theta_0)\partial_y f\left(t, y,x ; \theta\right)\\
h_{j2}(t,y,x; \theta) &= b(y; \theta_0)\partial_y h_j(t,y,x, \theta)
\end{align*}
satisfy (\ref{eqn1new}) for $j=1,2$.
Then, for each $m\in \NN$, there exists a constant $C_{m,K}>0$, such that 
\begin{align*}
\EE_{\theta_0}\left| \widetilde{F}_n(\theta) -
  \widetilde{F}_n(\theta')\right|^{2m} &\leq C_{m,K}\, |\theta-\theta '|^{2m} 
\end{align*}
for all $\theta, \theta' \in K$ and $n\in \NN$.
\theoqed
\label{lemma_moment_diff}
\end{lemma}
\begin{proof}[\textbf{Proof of Lemma \ref{lemma_moment_diff}}]
For use in the following, define, in addition to $h_1$, $h_2$ and
$h_{j2}$, the functions
\begin{align*}
h_{j1}(t,y,x; \theta) &= \partial_t h_j(t,y,x;\theta) + a(y)\partial_y
h_j(t,y,x;\theta) + \tfrac{1}{2} b^2(y; \theta_0) \partial_y^2
h_j(t,y,x;\theta) \\
h_{j21}(t,y,x; \theta) &= \partial_t h_{j2}(t,y,x; \theta)
+ a(y)\partial_y h_{j2}(t,y,x; \theta) + \tfrac{1}{2}  b^2(y;
\theta_0)\partial^2_y h_{j2}(t,y,x; \theta) \\
h_{j22}(t,y,x; \theta) &=  b(y; \theta_0)\partial_y h_{j2}(t,y,x; \theta)
\end{align*}
for $j=1,2$, and, for ease of notation, let
\begin{align*}
H^{n,i}_j(u; \theta, \theta') = Dh_j(u-t_{i-1}^n, X_u, \xtl;
\theta, \theta')
\end{align*}
for $j \in \{1,2,11,12,21,22,121,122,221,222\}$, where $Dk(\,\cdot\,; \theta,
\theta') = k(\,\cdot\,; \theta) - k(\,\cdot\,; \theta')$. Recall that $\Delta_n = 1/n$.
\medskip

First, by the martingale
properties of 
\begin{align*}
\Delta_n \sum_{i=1}^n \int_0^r \mathbf{1}_{( t_{i-1}^n,
  t_i^n]}(u) H^{n,i}_2(u; \theta, \theta')\, dW_u\,,
\end{align*}
the Burk\-hol\-der-Da\-vis-Gun\-dy inequality is used to
establish the existence of a constant $C_{m}> 0$ such that
\begin{align*}
\EE_{\theta_0} \left( \left| \Delta_n \sum_{i=1}^n
    \int_{t_{i-1}^n}^{t_i^n} H^{n,i}_2(u; \theta, \theta')\, dW_u\right|^{2m}\right) &\leq C_{m} \EE_{\theta_0}\left( \left|\Delta_n^2 \sum_{i=1}^n
  \int_{t_{i-1}^n}^{t_i^n} H^{n,i}_2(u; \theta, \theta')^2\, du\right|^m\right)\,.
\end{align*}
Now, using also Ito's formula, Jensen's inequality and Lemma \ref{triv_ineq},
\begin{align}
&\hspace{5mm} \EE_{\theta_0}\left( \left| \Delta_n \sum_{i=1}^n Df(\Delta_n, \xtr,
    \xtl; \theta, \theta')\right|^{2m}\right)\nonumber\\
&\leq C_m\EE_{\theta_0}\left( \left| \Delta_n \sum_{i=1}^n
      \int_{t_{i-1}^n}^{t_i^n} H^{n,i}_1(u; \theta, \theta') \, du \right|^{2m}\right) + C_m \EE_{\theta_0}\left( \left| \Delta_n \sum_{i=1}^n
      \int_{t_{i-1}^n}^{t_i^n} H^{n,i}_2(u; \theta,\theta') \, dW_u
    \right|^{2m}\right) \nonumber\\
&\leq C_{m} \Delta_n\sum_{i=1}^n
  \EE_{\theta_0}\left(\left|\int_{t_{i-1}^n}^{t_i^n} H^{n,i}_1(u; \theta,\theta') \, du
    \right|^{2m}  \right)  + C_m \EE_{\theta_0}\left( \left|\Delta_n^2 \sum_{i=1}^n
  \int_{t_{i-1}^n}^{t_i^n} H^{n,i}_2(u; \theta, \theta')^2\, du\right|^m\right) \nonumber\\
&\leq C_{m}\Delta_n^{2m+1}\sum_{i=1}^n \left(  
  \EE_{\theta_0}\left(\left|\frac{1}{\Delta_n}\int_{t_{i-1}^n}^{t_i^n}
  H^{n,i}_1(u; \theta, \theta') \, du
    \right|^{2m}  \right)  + \EE_{\theta_0}\left( \left|
  \frac{1}{\Delta_n}\int_{t_{i-1}^n}^{t_i^n} H^{n,i}_2(u; \theta,\theta')^2\,
  du\right|^m\right)\right) \nonumber\\
&\leq C_{m}\Delta_n^{2m}\sum_{i=1}^n \left( \int_{t_{i-1}^n}^{t_i^n}
  \EE_{\theta_0}\left( |H^{n,i}_1(u; \theta,\theta')|^{2m}\right) \, du +
  \int_{t_{i-1}^n}^{t_i^n} \EE_{\theta_0}\left( |H^{n,i}_2(u; \theta,\theta')|^{2m}\right)\,
  du\right) \label{eqn:H1H2} \\
&\leq C_{m,K}|\theta-\theta'|^{2m}\Delta_n^{2m} \nonumber\,,
\end{align}
thus
\begin{align*}
\EE_{\theta_0}\left( |DF_n(\theta, \theta')|^{2m} \right) 
&= \Delta_n^{-2m} \EE_{\theta_0}\left(  \left| \Delta_n\sum_{i=1}^n Df(\Delta_n, \xtr,
  \xtl; \theta, \theta')\right|^{2m}\right) \leq C_{m,K} |\theta-\theta'|^{2m}
\end{align*}
for all $\theta, \theta' \in K$ and $n\in \NN$. In the case where also $h_j$ and $h_{j2}$ satisfy
(\ref{eqn1new}) for all $x\in\xx$, $\theta \in \Theta$ and $j=1,2$, use Lemma \ref{triv_ineq} to write
\begin{align*}
&\hspace{-5mm} \EE_{\theta_0}\left( |H^{n,i}_1(u; \theta, \theta')|^{2m}\right)\\
 &\leq
C_{m}(u-t_{i-1}^n)^{2m-1} \int_{t_{i-1}^n}^u \EE_{\theta_0}\left(
  |H^{n,i}_{11}(v; \theta, \theta')|^{2m}\right)\, dv \\
&\hspace{5mm}
+ C_{m} (u-t_{i-1}^n)^{m-1} \int_{t_{i-1}^n}^u \EE_{\theta_0}\left(
  |H^{n,i}_{12}(v; \theta,\theta')|^{2m}\right)\, dv \\
&\leq C_{m}(u-t_{i-1}^n)^{2m-1} \int_{t_{i-1}^n}^u \EE_{\theta_0}\left(
  |H^{n,i}_{11}(v; \theta, \theta')|^{2m}\right)\, dv \\
&\hspace{5mm} +  C_{m} 
(u-t_{i-1}^n)^{m-1} \int_{t_{i-1}^n}^u \left(
    (v-t_{i-1}^n)^{2m-1}\int_{t_{i-1}^n}^v
  \EE_{\theta_0}\left(\left|H^{n,i}_{121}(w; \theta,\theta')\right|^{2m}\right)\,
  dw\right) \, dv \\
&\hspace{5mm} +C_{m} (u-t_{i-1}^n)^{m-1} \int_{t_{i-1}^n}^u \left(
    (v-t_{i-1}^n)^{m-1}\int_{t_{i-1}^n}^v
  \EE_{\theta_0}\left(\left|H^{n,i}_{122}(w; \theta,\theta')\right|^{2m}\right)\,
  dw\right) \, dv \\
&\leq C_{m,K} |\theta-\theta'|^{2m} \left( (u-t_{i-1}^n)^{2m} + (u-t_{i-1}^n)^{3m} \right) \,,
\end{align*}
and similarly obtain
\begin{align*}
\EE_{\theta_0}\left( |H^{n,i}_2(u; \theta,\theta')|^{2m}\right)
 &\leq  C_{m,K} |\theta-\theta'|^{2m} \left( (u-t_{i-1}^n)^{2m}
 +(u-t_{i-1}^n)^{3m} \right)\,.
\end{align*}
Now, inserting into (\ref{eqn:H1H2}),
\begin{align*}
&\hspace{-5mm} \EE_{\theta_0}\left( \left| \Delta_n \sum_{i=1}^n Df(\Delta_n, \xtr,
    \xtl; \theta, \theta')\right|^{2m}\right)\\
&\leq C_{m,K}\Delta_n^{2m}\sum_{i=1}^n \left( \int_{t_{i-1}^n}^{t_i^n}
  \EE_{\theta_0}\left( |H^{n,i}_1(u; \theta,\theta')|^{2m}\right) \, du +
  \int_{t_{i-1}^n}^{t_i^n} \EE_{\theta_0}\left( |H^{n,i}_2(u; \theta,\theta')|^{2m}\right)\,
  du\right) \\
&\leq C_{m,K} |\theta-\theta'|^{2m}\Delta_n^{2m}   \sum_{i=1}^n
\int_{t_{i-1}^n}^{t_i^n} \left( (u-t_{i-1}^n)^{2m} +(u-t_{i-1}^n)^{3m} \right)\,
  du \\
&\leq C_{m,K} |\theta-\theta'|^{2m}\left(  \Delta_n^{4m}  + \Delta_n^{5m} \right)\,,
\end{align*}
and, ultimately,
\begin{align*}
\EE_{\theta_0}\left( |D\tilde{F}_n(\theta, \theta')|^{2m} \right) &=\EE_{\theta_0}\left(  \left| \Delta_n^{-1}\sum_{i=1}^n Df(\Delta_n, \xtr,
  \xtl; \theta, \theta')\right|^{2m}\right) \\
&= \Delta_n^{-4m} \EE_{\theta_0}\left(  \left| \Delta_n\sum_{i=1}^n Df(\Delta_n, \xtr,
  \xtl; \theta, \theta')\right|^{2m}\right) \\
&\leq C_{m,K} |\theta-\theta'|^{2m}\left(  1  + \Delta_n \right) \\
&\leq C_{m,K} |\theta-\theta'|^{2m}\,.\qedhere
\end{align*}
\end{proof}
\begin{lemma}
Suppose that Assumption \ref{assumptions_on_X} is satisfied. Let $f
\in \cc_{0,1}^{\text{pol}}\left( \xx \times \Theta\right) $. Define
\begin{align*}
F(\theta) &= \int_0^1 f(X_s; \theta)\, ds
\end{align*}
and let $K\subseteq \Theta$ be compact and convex. Then, for each $m\in
\NN$, there exists a constant $C_{m,K} > 0$ such that for all $\theta, \theta' \in K$,
\begin{align*}
\EE_{\theta_0}|F(\theta) -
F(\theta')|^{2m} \leq C_{m,K} \,\, |\theta-\theta'|^{2m}\,.
\end{align*}
\theoqed
\label{lemma_limit_moment_diff}
\end{lemma}
Lemma
\ref{lemma_limit_moment_diff} follows from a simple application of the
mean value theorem.
\begin{lemma}
Let $K \subseteq \Theta$ be compact. 
Suppose that $\mathbf{H}_n = (H_n(\theta))_{\theta \in K}$
defines a sequence $(\mathbf{H}_n)_{n\in\NN}$ of continuous,
real-valued stochastic processes such that
\begin{align*}
H_n(\theta) \overset{\pp}{\longrightarrow} 0
\end{align*}
point-wise for all $\theta \in K$. Furthermore, assume that for some $m\in \NN$, there exists a constant
$C_{m,K}>0$ such that for all $\theta, \theta ' \in K$ and $n\in \NN$,
\begin{align}
\EE_{\theta_0}\left| H_n(\theta) -
  H_n(\theta')\right|^{2m} &\leq C_{m,K} |\theta-\theta '|^{2m}\,.
\label{mean_eq_orig}
\end{align}
Then,
\begin{align*}
\sup_{\theta \in K} \left|H_n(\theta)\right| \overset{\pp}{\longrightarrow} 0\,.
\end{align*}
\theoqed
\label{uniform_convergence}
\end{lemma}
\begin{proof}[\textbf{Proof of Lemma \ref{uniform_convergence}}]
$(H_n(\theta))_{n \in \NN}$ is tight in $\RR$ for all
$\theta \in K$, so, using
(\ref{mean_eq_orig}), it follows from \citet[Corollary
16.9 \& Theorem 16.3]{Kallenberg97} that the sequence of processes $(\mathbf{H}_n)_{n \in \NN}$ is tight
in $\cc(K, \RR)$, the space of continuous (and bounded) real-valued
functions on $K$, and thus relatively compact in distribution. Also, for all
$d\in \NN$ and $(\theta_1, \ldots, \theta_d) \in K^d$,
\begin{align*}
\begin{pmatrix} H_n(\theta_1) \\ \vdots \\ H_n(\theta_d)\end{pmatrix}&
\overset{\dd}{\longrightarrow} \begin{pmatrix} 0 \\ \vdots \\
  0 \end{pmatrix}\,,
\end{align*}
so by \citet[Lemma
16.2]{Kallenberg97}, $\mathbf{H}_n \overset{\dd}{\longrightarrow} 0$
in $\cc(K,\RR)$ equipped with the uniform metric. Finally, by the
continuous mapping theorem, 
$\sup_{\theta \in K} |H_n(\theta)| \overset{\dd}{\longrightarrow} 0\,$,
%
and the desired result follows.
\end{proof}
\section*{Acknowledgement} 

We are grateful to the referees for their insightful comments and
suggestions that have improved the paper. Nina Munkholt Jakobsen was
supported by the Danish Council for 
Independent Research | Natural Science through a grant to Susanne Ditlevsen.
Michael S\o rensen was supported by the Center for Research in
Econometric Analysis of Time Series funded by the Danish National
Research Foundation. The research is part of the
Dynamical Systems Interdisciplinary Network
funded by the University of Copenhagen Programme of Excellence.
\bibliographystyle{apalike}
\bibliography{bibliography} 
\end{document}